\documentclass[UKenglish]{lipics}
\usepackage{hanoi}
\def\lt{{ Levitating Tower}}

\def\alphalt{{$\alpha$-\lt}}

\def\bouncing{{Bouncing}}
\def\bouncingt{{\bouncing\ Tower}}
\def\bouncingtpb{{\bouncingt\ Problem}}
\newcommand{\btextractingpoint}[1]{{\lfloor \frac{ #1}{2} \rfloor +1}}
\newcommand{\btinsertingpoint}[1]{{\lfloor \frac{#1+1}{2} \rfloor}}

\usepackage{version}
\excludeversion{INUTILE}
\excludeversion{TODO}
\excludeversion{SHORT}
\includeversion{LONG}
\includeversion{VLONG}
\newenvironment{lproof}{\begin{proof}}{\end{proof}}
\includeversion{PRACTICAL}
\excludeversion{PREVIOUSWORK}
\excludeversion{FUTUREWORK}
\includeversion{ACKNOWLEDGEMENTS}
\includeversion{DISKPILEPROBLEM}

\title{Bouncing Towers move faster than {\Hanoi} Towers, but still require exponential time}
\titlerunning{Bouncing Towers} %optional, in case that the title is too long; the running title should fit into the top page column

\author{J\'er\'emy Barbay}
\affil{
    Departamento de Ciencias de la Computaci\'on (DCC), \\
    Universidad de Chile,\\
    Santiago, Chile \\
  \texttt{jeremy@barbay.cl}}
\authorrunning{J. Barbay} %mandatory. First: Use abbreviated first/middle names. Second (only in severe cases): Use first author plus 'et. al.'

\Copyright{J\'er\'emy Barbay}%mandatory, please use full first names. LIPIcs license is "CC-BY";  http://creativecommons.org/licenses/by/3.0/

% mandatory: Please provide 1-5 keywords
\begin{document}

\maketitle

\begin{abstract}
The problem of the {Hano{\"\i}} Tower is a classic exercise in recursive programming: the solution has a simple recursive definition, and its complexity and the matching lower bound are the solution of a simple recursive function (the solution is so easy that most students memorize it and regurgitate it at exams without truly understanding it). We describe how some very minor changes in the rules of the {Hano{\"\i}} Tower yield various increases of complexity in the solution, so that they require a deeper analysis than the classical {Hano{\"\i}} Tower problem while still yielding exponential solutions.  In particular, we analyze the problem fo the Bouncing Tower, where just changing the insertion and extraction position from the top to the middle of the tower results in a surprising increase of complexity in the solution: such a tower of $n$ disks can be optimally moved in $\sqrt{3}^n$ moves for $n$ even (i.e. less than a {Hano{\"\i}} Tower of same height), via $5$ recursive functions (or, equivalently, one recursion function with $5$ states).
\end{abstract}

\newcommand{\move}[5]{\ensuremath{\mbox{\tt move#1}(#2,#3,#4,#5)}}
\newcommand{\unitarymove}[2]{\ensuremath{\mbox{\tt move}(#1\rightarrow#2)}}

\newcommand\IH{{I\hspace{-3pt}H}}

\newcommand{\nbdisksabove}[1]{{\lfloor \alpha #1 \rfloor}}
\newcommand{\extractingpoint}[1]{{\lfloor \alpha #1 \rfloor +1}}
\newcommand{\insertingpoint}[1]{{\lfloor \alpha (#1+1) \rfloor}}

\newcommand{\pegA}{\ensuremath A}
\newcommand{\pegB}{\ensuremath B}
\newcommand{\pegC}{\ensuremath C}
\newcommand{\pegX}{\ensuremath X}

\section{Introduction}

The \hanoitpb\ is a classical problem often used to teach recursivity, originally proposed in 1883 by \'Edouard Lucas~\cite{1883-MISC-LaTourDHanoi-Lucas,1883-BOOK-RecreationsMathematiques-Lucas}, where one must move $n$ disks, all of distinct size, one by one, from a peg $\pegA$ to a peg $\pegC$ using only an intermediary peg $\pegB$, while ensuring that at no time does a disk stands on a smaller one. As early as 1892, Ball~\cite{1892-BOOK-MathematicalRecreationsAndEssays-Ball} described an optimal recursive algorithm which moves the $n$ disks of a \hanoit\ in $2^n-1$ steps.
\begin{LONG}
Many generalizations have been studied, allowing more than three pegs~\cite{1941-AmericanMathematics-SolutionOfProblemNo2918-FrameStewart}, coloring disks~\cite{1985-JRM-TheTowersOfBrahmaAndHanoiRevisited-Wood}, and cyclic \hanoit s~\cite{1981-IPL-TheCyclicTowersOfHanoi-Atkison}.  Some problems are still open, as the optimality of the algorithm for $4$-peg \hanoitpb, and the analysis of the original problem is still a source of inspiration hundreds of year after its definition: for instance, Allouche and Dress~\cite{1990-RAIRO-ToursDeHanoiEtAutomates-AlloucheDress} proved in 1990 that the movements of the \hanoitpb\ can be generated by a finite automaton, making this problem an element of $SPACE(1)$.
\end{LONG}

%\begin{LONG}
The solution to the \hanoitpb\ is simple enough that it can be memorized and regurgitated at will by students from all over the world: asking about it in an assignment or exam does not truly test a student's mastery of the concept of recursivity, pushing instructors to consider variants with slightly more sophisticated solutions. Some variants do not make the problem more difficult (e.g. changing the \texttt{insertion} and \texttt{removal} point to the bottom: the solution is exactly the same), some make it only slightly more difficult (e.g. considering the case where the disks are not necessarily of distinct sizes\begin{DISKPILEPROBLEM}, described and analized in Appendix~\ref{sec:diskPileProblem})\end{DISKPILEPROBLEM}, but some small changes can make it surprisingly more difficult.
%\end{LONG}

We consider the \bouncingtpb, which only difference with the \hanoitpb\ is the \texttt{insertion} and \texttt{removal} point in each tower, taken to be the middle instead of the top (see Figure~\ref{figureExplicativeToursARessorts} for an illustration with \bouncingt s of sizes $n=3$ and $n=4$, and Section~\ref{sec:definition} for the formal definition). If the disks all weight the same, one can imagine such a tower as standing on a spring, the elasticity $k$ of the spring being tuned so that the middle of the tower is always at the same height, where disks are inserted and removed.

\begin{TODO}
Add a figure with a bouncing tower on a spring, with a red arrow going down for gravity and a green arrow with the resistance of the spring $k=g$ or $k=g/2$.
\end{TODO}

\begin {figure}[h]
\parbox{.45\textwidth}{\includegraphics[width=.45\textwidth]{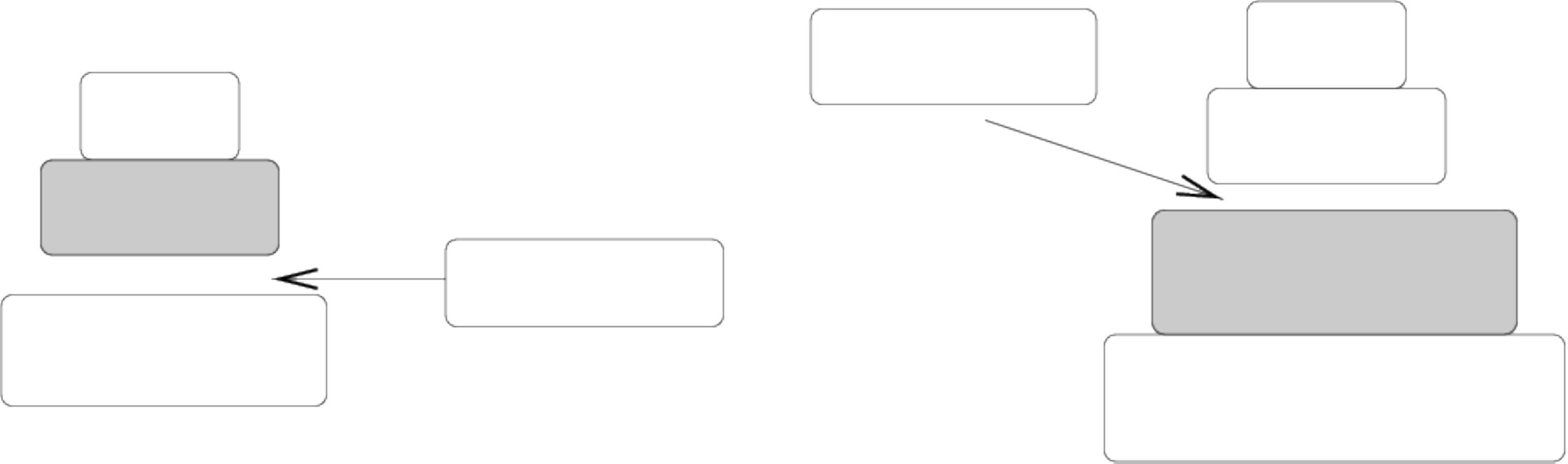}    }
\hfill
\parbox{.45\textwidth}{ \caption {An illustration of the rules for the \texttt{insertion} and \texttt{removal} in a \bouncingt, depending on the parity of its size (sizes $n=3$ and $n=4$ here).  In each case, the shaded disk indicates the \texttt{removal} point and the arrow indicates the \texttt{insertion} point.  \label {figureExplicativeToursARessorts} } }
\end {figure}

As for the classical \hanoit, such \texttt{insertion} and \texttt{removal} rules guarantee that any move is \emph{reversible} (i.e. any disk $d$ removed from a peg $X$ can always be immediately reinserted in the same peg $X$), that the \texttt{insertion} and \texttt{removal} positions are uniquely defined, that each peg can always receive a disk, and that each tower with one disk or more can always yield one disk. The problem is very similar to the \hanoitpb: one would expect answering the following questions to be relatively easy, possibly by extending the answers to the corresponding questions on \hanoit s\footnote{For a \hanoit, the answer to those question is that there is a single such shortest sequence, of length $2^n-1$, obtained by the recursion $h(n,A,B,C)=h(n{-}1,A,C,B)."A\rightarrow B;".h(n{-}1,B,A,C)$ if $n>0$ and $\emptyset$ otherwise.}:

\begin{quote}
\bf 
Consider the problem of moving a \bouncingt\ of $n$ disks, all of distinct size, one by one, from a peg $\pegA$ to a peg $\pegC$ using only an intermediary peg $\pegB$, while ensuring that at no time does a disk stands on a smaller one:
\begin{enumerate}
\item Which sequences of steps permit to move such a tower?
\item What is the minimal length of such a sequence?
\item How many shortest such sequences are there?
\end{enumerate}
\end{quote}

We show that there is a unique shortest sequence of steps which moves a \bouncingt\ of $n$ disks of distinct sizes, and that it is of length at most $\sqrt{3}^n$\begin{LONG} (i.e. exactly $\sqrt{3}^n=3^{\frac{n}{2}}$ if $n$ is even, and $\frac{3}{5}\sqrt{3}^{n-1}-\frac{2}{3}=3^{\frac{n-1}{2}}+2(3^{\frac{n-3}{2}}-1)<\sqrt{3}^n$ if $n$ is odd)\end{LONG}. As $\sqrt{3}\approx1.733<2$, this sequence is exponentially shorter than the corresponding one for the \hanoitpb\ (of length $2^n-1$).  We define formally the problem and its basic properties in Section~\ref{sec:form-defin-basic}\begin{LONG}: its formal definition in Section~\ref{sec:definition}, some examples where such towers can be moved faster in Section~\ref{sec:moving-small-towers}, and some useful concepts on the \texttt{insertion} and \texttt{removal} order of a tower in Section~\ref{sec:struct-facts-single}\end{LONG}.  We describe a recursive solution in Section~\ref{sec:solution}\begin{LONG}, via its algorithm in Section~\ref{sec:algorithm}, the proof of its correctness in Section~\ref{sec:corr-algor} and the analysis of its complexity in Section~\ref{sec:complexity-algorithm}\end{LONG}. The optimality of the solution is proved in Section~\ref{sec:optimality}, via an analysis of the graph of all possible states and transition (defined and illustrated in Section~\ref{sec:configuration-graph}) and a proof of optimality for each function composing the solution (Section~\ref{sec:proof-optimality}). We conclude with a discussion (Section~\ref{sec:discussion}) of various other variants of similar or increased complexity\begin{DISKPILEPROBLEM}, and share in Appendix~\ref{sec:diskPileProblem} the text and the solution of a simpler variant successfully used in undergraduate assignments and exams\end{DISKPILEPROBLEM}.
\begin{SHORT}
The proofs of correctness (Theorem~\ref{res:correctness}) and optimality (Theorem~\ref{res:proof-optimality}) are simple case studies, so they are described in Appendix~\ref{sec:omitted-proofs} for lack of space.
\end{SHORT}

\section{Formal Definition and Basic Facts}\label{sec:form-defin-basic}

In this section we define more formally the \bouncingt\ (Section~\ref{sec:definition}), how small examples already show that moving such towers require less steps than moving a \hanoit\ (Section~\ref{sec:moving-small-towers}), and some properties of the order in which disks are inserted or removed on a peg to build or destroy a tower (Section~\ref{sec:struct-facts-single}).

\subsection{Formal Definition}
\label{sec:definition}

The ``middle'' disk of a tower of even size is not well defined, nor is the ``middle'' \texttt{insertion} point in a tower of odd size: we define both more formally in such a way that if $n$ is odd, the \texttt{removal} position is the center one, and the \texttt{insertion} point is below it; while if $n$ is even, the \texttt{insertion} point is in the middle of the tower, while the \texttt{removal} position is below the middle of the tower (see Figure~\ref{figureExplicativeToursARessorts} for an illustration with sizes $n=3$ and $n=4$).
\begin{LONG}
More formally, on a peg containing $n$ disks ranked by increasing sizes,
the \texttt{removal} point is the disk of rank $\btextractingpoint{n}$; and
the \texttt{insertion} point is position $\btinsertingpoint{n}$.
  \end{LONG}

The \texttt{insertion} of disk $d$ on peg $X$ is {\em legal} if inserting $d$ in the \texttt{insertion} point of $X$ yields a legal configuration, where no disk is above a smaller one.  A move from peg $X$ to peg $Y$ is \emph{legal} if there is a disk $d$ to remove from $X$, and if the \texttt{insertion} of $d$ on the $Y$ is legal.

\subsection{Moving small towers - differences with \hanoi}\label{sec:moving-small-towers}

For size one or two, there is no difference in the moving cost between a \hanoit\ and a \bouncingt.  The first difference appears for size three, when only five steps are necessary to move a \bouncingt\ (see the sequence of five steps to move a \bouncingt\ of size $n=3$ in Figure~\ref{3diskBouncingTowerInMove}) as opposed to the seven steps required for moving a classical \hanoit\ (see the sequence of seven steps to move a \hanoit\ of size $n=3$ in Figure~\ref{3diskHanoiTowerInMove}).

\newcommand{\hanoiStateWithTransition}[4]{%
  \begin{tabular}{c}
   \hspace{-3cm} #4    \\
  \framebox[66pt]{\vbox{%
      \hbox{%
        \makebox[22pt]{#1}%
        \makebox[22pt]{#2}%
        \makebox[22pt]{#3}%
        }%
      \hbox{%
        \makebox[22pt]{$\pegA$}%
        \makebox[22pt]{$\pegB$}%
        \makebox[22pt]{$\pegC$}%
        }%
      }  
    } 
  \end{tabular}
  }%

\begin{figure}[h]
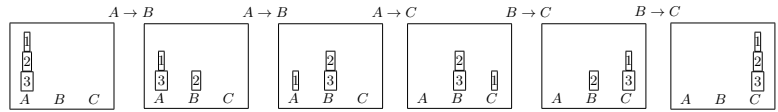

\centering
\scalebox{.6}{
  \hanoiStateWithTransition{\peg{\trois}{\deux}{\un}{}} {} {} {\ } %
  \hanoiStateWithTransition{\peg{\trois}{\un}{}{}} {\peg{\deux}{}{}{}} {}{$\pegA \rightarrow \pegB$} %
  \hanoiStateWithTransition{\peg{\un}{}{}{}} {\peg{\trois}{\deux}{}{}} {}{$\pegA \rightarrow \pegB$}% 
  \hanoiStateWithTransition{\peg{}{}{}{}} {\peg{\trois}{\deux}{}{}} {\peg{\un}{}{}{}}{$\pegA \rightarrow \pegC$} %
  \hanoiStateWithTransition{}{\peg{\deux}{}{}{}}{\peg{\trois}{\un}{}{}}{$\pegB \rightarrow \pegC$}%
  \hanoiStateWithTransition{}{}{\peg{\trois}{\deux}{\un}{}}{$\pegB \rightarrow \pegC$}%
}
\caption{A \bouncingt\ of three disks can be moved in just five steps.\label{3diskBouncingTowerInMove}}

\end{figure}
\begin{figure}[h]
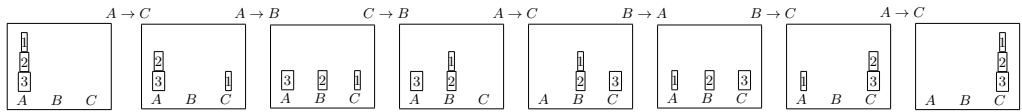

\centering
\scalebox{.6}{
  \hanoiStateWithTransition{\peg{\trois}{\deux}{\un}{}} {} {}  {\ }  %
  \hanoiStateWithTransition{\peg{\trois}{\deux}{}{}} {\peg{}{}{}{}} {\peg{\un}{}{}{}} {$\pegA \rightarrow \pegC$}%
  \hanoiStateWithTransition{\peg{\trois}{}{}{}} {\peg{\deux}{}{}{}} {\peg{\un}{}{}{}} {$\pegA \rightarrow \pegB$}%
  \hanoiStateWithTransition{\peg{\trois}{}{}{}} {\peg{\deux}{\un}{}{}} {} {$\pegC \rightarrow \pegB$}%
  \hanoiStateWithTransition{} {\peg{\deux}{\un}{}{}} {\peg{\trois}{}{}{}} {$\pegA \rightarrow \pegC$}%
  \hanoiStateWithTransition{\peg{\un}{}{}{}} {\peg{\deux}{}{}{}} {\peg{\trois}{}{}{}} {$\pegB \rightarrow \pegA$}%
  \hanoiStateWithTransition{\peg{\un}{}{}{}} {} {\peg{\trois}{\deux}{}{}} {$\pegB \rightarrow \pegC$}%
  \hanoiStateWithTransition{}{}{\peg{\trois}{\deux}{\un}{}} {$\pegA \rightarrow \pegC$}%
}
\caption{A \hanoit\ of three disks require seven steps to be moved between two pegs.\label{3diskHanoiTowerInMove}}
\end{figure}

When an odd number of disks is present on the peg $\pegA$, and an even number is present on pegs $\pegB$ and $\pegC$, a sub-tower of height $2$ can be moved from $A$ in $2$ steps, when in a \hanoit\ we need $3$ steps to move any subtower of same height.  In the \bouncingtpb, having a third disk ``fixed'' on $\pegA$ yields a reduced number of steps. We formalize this notion of ``fixed'' disk in the next section.

\subsection{Structural facts on a single Peg}\label{sec:struct-facts-single}

Before considering the complete problem over three pegs, we describe some concept about single pegs, and on the order in which the disks are inserted and removed on a specific peg.

\begin{definition}
We define the {\em removal order} as the order in which disks (identified by their rank in the final tower) can be removed from a \bouncingt.  Symmetrically, we define the {\em insertion order} as the order in which the disks are inserted in the tower.
\end{definition}

The symmetry of the rules concerning the \texttt{insertion} and \texttt{removal} location of \bouncingt s yields that the {\em insertion} order is the exact reverse of the {\em removal} order (the \texttt{insertion} point of a tower is the \texttt{removal} point of a tower with one more disk), and each disk removed from a peg can be immediately replaced exactly where it was.

In particular, a key argument to both the description of the solution in Section~\ref{sec:solution} and to the proof of its optimality in Section~\ref{sec:optimality} is the fact that, when some (more extreme) disks are considered as ``fixed'' (i.e. the call to the current function has to terminate before such disks are moved), the order in which a subset of the disks is removed from a peg depends on the number of those ``fixed'' disks.

\begin{TODO}
Add a figure with fixed disks
\end{TODO}

\begin{definition}
When moving recursively $n$ disks from a peg $\pegX$ with $x>n$ disks, the $x-n$ last disks in the \texttt{removal} order of $\pegX$ are said to be \emph{fixed}.  The {\em parity} of peg $\pegX$ is the parity of the number $x$ of disks {\em fixed} on this peg.
\end{definition}

\begin{LONG}
\bouncingt s cannot be moved much faster than \hanoit s:
\begin{lemma}\label{lemmaNoMoveWithSameParity}
It is impossible to move more than one disk between two pegs of same parity without a third peg.
\end{lemma}
\begin{lproof}
Between two pegs of same parity, the \texttt{removal} order is the same.  So the first disk needed on the final peg will be the last one removed from the starting peg.  With more than one disk, we need the third peg to dispose temporally other disks.
\end{lproof}

\begin{lemma}\label{lemmaNoMoveWithDistinctParity}
It is impossible to move more than two disk between two pegs of opposite parities without a third peg.
\end{lemma}
\begin{lproof}
Between two pegs of opposite parities, the \texttt{removal} orders are different: But the definition of the middle is constant when the number of disks changes of $2$. So after moving two disks the third cannot be inserted in the right place.
\end{lproof}
\end{LONG}

The \texttt{removal} and \texttt{insertion} orders are changing with the parity of the \bouncingt: Consider a peg with $n$ disks on it:
\begin{itemize}
\item if $n=2m+1$ is odd, then the disks are removed in the following order:
  $$ (m+1,   m+2,\\
   m,m+3,\\
   m-1,m+4,\\
   \ldots,\\
   3, 2m, \\
   2, 2m+1,\\
   1)
$$
\item if $n=2m$ is even, then the removal order is:
$$ (m+1, m,\\
m+2,m-1,\\
m+3,m-2,\\
\ldots,\\
2m-1,2,\\
2m,1)
$$
\end{itemize}
The relative order of $m$ and $m+2$, of $m-1$ and $m+3$, and more generally of any pair of disks $i$ and $m-i$ for $i\in[1..\lfloor n/2\rfloor]$, are distinct.  More specifically, disks are alternately extracted below and above the \texttt{insertion} point. This implies the two following connexity lemma:
\begin{lemma}\label{connexitylemma}
The $k$ first disks removed from the tower are contiguous in the original tower, and they are either all smaller or all larger than the $(k+1)$-th disk removed.
\end{lemma}
\begin{TODO}
PROVE This lemma!
\end{TODO}

\begin{lemma}\label{reciprocofconnexitylemma}
If $k$ disks are all smaller than the disk below the \texttt{insertion} point, and all larger than the disk above the \texttt{insertion} point, then there exists an order in which to add those $k$ disks to the tower.
\end{lemma}
\begin{lproof}
By induction: 
for one disk it is true; 
for $k$ disks, if the \texttt{insertion} point after the \texttt{insertion} of disc $d$ is
above $d$ then add the larger and then the $k-1$ disks left,
else add the smaller and then the $k-1$ disks left.
\end{lproof}

We present in the next section a solution to the \bouncingtpb\ which takes advantage of the cases where two disks can be moved between the same two pegs in two consecutive steps.

\section{Solution}\label{sec:solution}

One important difference between \hanoit s and \bouncingt s is that we need not always to remove $n-1$ disks of a tower of $n$ disks to place the $n$-th disk on another peg (e.g. in the sequence of steps shown in Figure~\ref{3diskBouncingTowerInMove}, disk $3$ was removed from $A$ when there was still a disk sitting on top of it).  But we need always to remove at least ${n-2}$ disks in order to release the $n$-th disk, as it is the last or the last-but-one disk removed.  This yields a slightly more complex recursion than in the traditional case.  We describe an algorithmic solution in Section~\ref{sec:algorithm},  prove its correctness in Section~\ref{sec:corr-algor}, and analyze the length of its output in Section~\ref{sec:complexity-algorithm}. We prove the optimality of the solution produced separately, in Section~\ref{sec:optimality}.

\subsection{Algorithm} \label{sec:algorithm}
\providecommand{\move}[5]{\ensuremath{\mbox{\tt move#1}(#2,#3,#4,#5)}}
\providecommand{\unitarymove}[2]{\ensuremath{\mbox{\tt move}(#1,#2)}}

Note $|\pegA|$ the number of disks on peg $\pegA$, $|\pegB|$ on $\pegB$ and $|\pegC|$ on $\pegC$.  For each triplet $(x,y,z)\in\{0,1\}^3$, we define the function $\move{xyz}{n}{A}{B}{C}$ moving $n$ disks from peg $\pegA$ to peg $\pegC$ using peg $\pegB$ when $|A|\geq n$, $|A|-n \equiv x \mod 2$, $|B| \equiv y \mod 2$, $|C| \equiv z \mod 2$, and the $n$ first disks extracted from $A$ can be legally inserted on $B$ and $C$. Less formally, there are $x$ fixed disks on the peg $A$, $y$ on $B$ and $z$ on $C$.

\begin{TODO}
DEFINE $(p)_{xyz}$ from a previous version.
\end{TODO}

We need only to study three of those $2^3=8$ functions. 
First, as the functions are symmetric two by two: for instance, $\move{000}{n}{A}{B}{C}$ behaves as $\move{111}{n}{A}{B}{C}$ would if the \texttt{insertion} point in a tower of odd size was above the middle disk, and the \texttt{removal} point in a tower of even size was above the middle of the tower: in particular, they have exactly the same complexity. 
Second, the reversibility and symmetry of the functions yields a similar reduction: $\move{001}{n}{A}{B}{C}$ has the same structure as the function $\move{100}{n}{A}{B}{C}$ and the two have the same complexity.

We describe the python code implementing those functions in Figures~\ref{fig:move000}to~\ref{fig:move010}, so that the initial call is made through the call \verb+move000(n,"a","b","c")+, while recursive calls refer only to functions $\move{000}{n}{A}{B}{C}$ (Figure~\ref{fig:move000}), $\move{100}{n}{A}{B}{C}$ (Figure~\ref{fig:move100}), $\move{001}{n}{A}{B}{C}$ (similar to $\move{100}{n}{A}{B}{C}$ and described \begin{SHORT}in the Appendix \end{SHORT} in Figure~\ref{fig:move001}) and $\move{010}{n}{A}{B}{C}$ (Figure~\ref{fig:move010}).

\begin{figure}
\centering
\begin{minipage}[t]{.32\linewidth}
\caption{\\$move000(n,A,B,C)$\label{fig:move000}}
\begin{lstlisting}
def move(a,b):
  print "("+a+",",
  print b+")",
     
def move000(n,a,b,c):
  if n>0 :
    move100(n-1,a,c,b)
    move(a,c)
    move001(n-1,b,a,c)
\end{lstlisting}
\end{minipage} \hfill
\begin{minipage}[t]{.32\linewidth}
\caption{\\$move100(n,A,B,C)$\label{fig:move100} \begin{SHORT}($move001(n,A,B,C)$ is defined similarly in Figure~\ref{fig:move001} in the Appendix)\end{SHORT}
}
\begin{lstlisting}
def move100(n,a,b,c):
  if n == 1 :
    move(a,c)
  elif n>1 :
    move100(n-2,a,c,b)
    move(a,c)
    move(a,c)
    move010(n-2,b,a,c)
\end{lstlisting}      
\begin{LONG}
\caption{\\$move001(n,A,B,C)$\label{fig:move001}}
\begin{lstlisting}
def move001(n,a,b,c):
  if n == 1 :
    move(a,c)
  elif n>1 :
    move010(n-2,a,c,b)
    move(a,c)
    move(a,c)
    move001(n-2,b,a,c)
\end{lstlisting}
\end{LONG}
\end{minipage}
\hfill
\begin{minipage}[t]{.32\linewidth}
\caption{\\$move010(n,A,B,C)$\label{fig:move010}}
\begin{lstlisting}
def move010(n,a,b,c):
  if n == 1 :
    move(a,c)
  elif n == 2 :
    move(a,b)
    move(a,c)
    move(b,c)
  elif n>2 :
    move010(n-2,a,b,c)
    move(a,b)
    move(a,b)
    move010(n-2,c,b,a)
    move(b,c)
    move(b,c)
    move010(n-2,a,b,c)
\end{lstlisting}
\end{minipage}
\label{fig:pythonCode}
\end{figure}

\begin{INUTILE}
\begin{figure}
\begin{minipage}[t]{.32\linewidth}
\begin{minipage}[t]{1\linewidth}
\caption{\\$move000(n,A,B,C)$\label{fig:move000}}
\begin{lstlisting}
% |A|-n >0 is even; 
% |B| is even;
% |C| is even. 
IF n>0
 move100(n-1,A,C,B)
 move(A,C)
 move001(n-1,B,A,C)
ENDIF
\end{lstlisting}
\end{minipage}
\begin{LONG}
\begin{minipage}[t]{1.0\linewidth}
\caption{\\$move100(n,A,B,C)$\label{fig:move100}}
\begin{lstlisting}
% |A|-n >0 is odd; 
% |B| is even;
% |C| is even.
IF n==1
 move (A,C);
ELSE
 move100(n-2,A,C,B);
 move(A,C);
 move(A,C);
 move010(n-2,B,A,C);
ENDIF
\end{lstlisting}
\end{minipage}
\end{LONG}
\end{minipage}
\begin{minipage}[t]{.32\linewidth}
\caption{\\$move001(n,A,B,C)$\label{fig:move001}}
\begin{lstlisting}
% |A|-n >0 is even; 
% |B| is even;
% |C| is odd.
IF n==1
 move(A,C);
ELSE
 move010(n-2,A,C,B);
 move(A,C);
 move(A,C);
 move001(n-2,B,A,C);
ENDIF
\end{lstlisting}
\end{minipage}
\hfill
\begin{minipage}[t]{.32\linewidth}
\caption{\\$move010(n,A,B,C)$\label{fig:move010}}
\begin{lstlisting}
% |A|-n >0 is even; 
% |B| is odd;
% |C| is even.
IF n==1
 move(A,C);
ELSIF n==2
 move(A,B);
 move(A,C);
 move(B,C);
ELSE
 move010(n-2,A,B,C);
 move(A,B);
 move(A,B);
 move010(n-2,C,B,A);
 move(B,C);
 move(B,C);
 move010(n-2,A,B,C);
ENDIF
\end{lstlisting}
\end{minipage}
\end{figure}
\end{INUTILE}

The algorithm for $\move{000}{n}{A}{B}{C}$ (in Figure~\ref{fig:move000}) has the same structure as the corresponding one for moving \hanoit s, the only difference being in the parity of the pegs in the recursive calls, which implies calling other functions than $\move{000}{n}{A}{B}{C}$, in this case $\move{001}{n}{A}{B}{C}$ and $\move{100}{n}{A}{B}{C}$.
The algorithms for  $\move{100}{n}{A}{B}{C}$ 
(in Figure~\ref{fig:move100}) and $\move{001}{n}{A}{B}{C}$ (in Figure~\ref{fig:move001}) and
are taking advantage of the difference of parity between the two extreme pegs to move two consecutive disks in two moves, but still has a similar structure to the algorithm for $\move{000}{n}{A}{B}{C}$ and the corresponding one for moving \hanoit s (just moving two disks instead of one).

The algorithm for $\move{010}{n}{A}{B}{C}$ is less intuitive.  
Given that the \texttt{removal} and \texttt{insertion} orders on the origin peg $A$ and on the destination peg $C$ are the same (because the parity of those pegs is the same), $n-1$ disks must be removed from $A$ before the last disk of the \texttt{removal} order\begin{LONG}, which yields a naive algorithm such as described in Figure~\ref{fig:NonOptimalMove010}\end{LONG}.
Such a strategy would yield a correct solution but not an optimal one, as it reduces the size only by one disk at the cost of two recursive calls and one step (i.e. reducing the size by two disks at the cost of four recursive calls and three steps), when another strategy (described in the algorithm in Figure~\ref{fig:move010}) reduces the size by two at the cost of three recursive calls and four steps\begin{LONG}: moving $n-2$ disks to $C$, the two last disks of the \texttt{removal} order on $B$, then $n-2$ disks to $A$, the two last disks of the \texttt{removal} order on $C$, then finally the $n-2$ disks to $C$\end{LONG}. 
The first strategy ($f(n)=2f(n-1)+2=4f(n-2)+3$) yields a complexity within $\Theta(2^n)$ while the second strategy ($f(n)=3f(n-2)+4$) yields a complexity within $\Theta(3^\frac{n}{2})$. We show in Section \ref{sec:corr-algor} that moving two disks at a time is correct in this context and in Section~\ref{sec:optimality} that the latter yields the optimal solution.

\begin{LONG}
\begin{figure}
\begin{minipage}[t]{.3\linewidth}
\caption{\\Alternative (non optimal) take on $move010(n,A,B,C)$\label{fig:NonOptimalMove010}}
\begin{lstlisting}
% |A|-n >0 is even; 
% |B| is odd;
% |C| is even.
IF n==1
 move(A,C);
ELSE
 move101(n-1,A,C,B);
 move(A,C);
 move101(n-1,B,A,C);
ENDIF
\end{lstlisting}
\end{minipage}
\hfill
\begin{minipage}[t]{.3\linewidth}
\caption{\\Alternative (non optimal) take on $move101(n,A,B,C)$\label{fig:NonOptimalMove101}}
\begin{lstlisting}
% |A|-n >0 is even; 
% |B| is odd;
% |C| is even.
IF n==1
 move(A,C);
ELSE
 move010(n-1,A,C,B);
 move(A,C);
 move010(n-1,B,A,C);
ENDIF
\end{lstlisting}
\end{minipage}
\end{figure}
\end{LONG}

\subsection{Correctness of the algorithm}\label{sec:corr-algor}
\providecommand\IH{H}

We prove the correctness of our solution by induction on the number $n$ of disks.

\begin{theorem}\label{res:correctness}
For any positive integer value $n$, 
and any triplet  $(x,y,z)\in\{0,1\}^3$ of booleans,
the function $\move{xyz}{n}{A}{B}{C}$ produces 
a sequence of legal steps which moves a \bouncingt\ from $A$ to $C$ via $B$.
\end{theorem}

The proof is based on the following invariant, satisfied by all recursive functions on entering and exiting:
\begin{definition}{\em Requirement for insertion $(i)$:}
The disks above the \texttt{insertion} point of $B$ or $C$ are all smaller than the first $n$ disks removed from $A$;
and the disks below the \texttt{insertion} point of $B$ or $C$ are all larger than the first $n$ disks removed from $A$  (see an illustration in Figure~\ref{requirementForInsertion}).
\end{definition} 

\begin{figure}[hbtf]
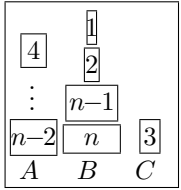

\parbox{4cm}{
\[ \hanoiState
{\peg{\nmdeux}{\diskDots}{\quatre}{}}
{\peg{\n}{\nmun}{\deux}{\un}}
{\peg{\trois}{}{}{}}
\]
} \parbox{9cm}{
  \caption{{\em Requirement for insertion $(i)$:}\label{requirementForInsertion} disks $4$ to $n-2$ can be inserted on $B$ as the \texttt{insertion} point of $B$ is between $2$ and~$n-1$; and on $C$ as the \texttt{insertion} point of $C$ is under~$3$.  }}
\end{figure}

\begin{TODO}
RECOVER definition of $p(100)$ the predicates before each function from previous version of paper.
\end{TODO}

\begin{lproof}
Consider the property $\IH(n)=$ ``$\forall(x,y,z)\in\{0,1\}^3,$ $\forall i \leq n,$ $\move{xyz}{i}{A}{B}{C}$ is correct''.  $\IH(0)$ is trivially true, and $\IH(1)$ can be checked for all functions at once. For all values $x,y,z$, the function $\move{xyz}{1}{A}{B}{C}$ is merely performing the step $\unitarymove{A}{C}$. The hypothesis $\IH(1)$ follows.  Now, for a fixed $n> 1$, assume that $\IH(n-1)$ holds: we prove the hypothesis $\IH(n)$ separately for each function.

\begin{itemize}
% Correctness of MOVE 000
\item {Analysis of $\move{000}{n}{A}{B}{C}$:}
\begin{enumerate}
\item According to $\IH(n-1)$ the call to $\move{100}{n-1}{A}{B}{C}$ is correct if $(i)$ and $(p)_{100}$ are respected.  $(i)$ is implied by $(i)$ on $\move{000}{n-1}{A}{B}{C}$; $(p)_{100}$ is implied by $(p)_{000}$ and the remaining disk on $A$ ($a-n \mod 2 \equiv 0 \Rightarrow a-(n-1)  \mod 2 \equiv 1 \mod 2$).
  % According to $\IH(n-1)$ the call to $\move{100}$ is correct, so it moves effectively $n$ disks from $A$ to $B$.  Note that precondition $(i)$ is satisfied by the requirement for $\move{000}$, and that the requirement $(p)_{100}$ on the parity of the number of disks on $A$ is satisfied because of the remaining disk on $A$, while $\IH(n-1)$ gives $(p)_{000}$.
\item The step $\unitarymove{A}{C}$ is possible and legal because of the precondition $(i)$ for $\move{000}{n}{A}{B}{C}$: the disk moved was in the $n$ first removed from $A$, and so can be introduced on $C$.
\item The call to $\move{001}{n}{A}{B}{C}$ is symmetrical to $1$, and so correct.
\item We can check the final state by verifying that the number of disks removed from $A$ and added to $C$ is $(n-1) + 1 = n$.
\end{enumerate}
So $\move{000}{n}{A}{B}{C}$ is correct.

% Correctness of MOVE 100
\item {Analysis of $\move{100}{n}{A}{B}{C}$:}
\begin{enumerate}
  \item $\move{100}{n-2}{A}{B}{C}$ is correct according to $\IH(n-1)$, as the
  requirements are also:
  The requirement $(i)$ is given by $(i)$ for the initial call, 
  and the parity $(p)_{100}$ is respected because we move two disks 
  less than in the current call to $\move{100}{n}{A}{B}{C}$.  
%  So this call $\move{100}(n-2)$ is correct and moves $n-2$ disks to $B$
\item The two disks left (let us call them $\alpha$ and $\beta$) are in position (given fig. \ref{TwoLastDisksRemoved}, $(i)$) such that the removal order on $A$ is $(\alpha,\beta)$ and the \texttt{insertion} order on $C$ is $(\beta,\alpha)$ (see fig.\ref{TwoLastDisksRemoved}, $(ii)$).  They can be inserted on $C$ because of requirement $(i)$.  So the two disks are correctly moved in two steps.

%  the first one extracted is the smallest (because the number of disks
%  on the peg is odd) and the second disk follows.
  
%   On peg $C$ after the insertion of the disk $\alpha$, the insertion
%   point is below as the number of disk is then odd: so we can insert
%   the disk $\beta$.  The two disks are well placed one compared to the
%   other, and also compared to other disks because of precondition
%   $(i)$.
  
\item The requirements for $\move{010}{n-2}{A}{B}{C}$ are satisfied:
  \begin{itemize}
\item{$(i)$} stand as a consequence of the precondition $(i)$ for the current call, as the $n-2$ disks to be moved on $C$ were on $A$ before the original call, in the middle of $\alpha$ and $\beta$.
\item{$(p)_{010}$}: The number of disks on $C$ is still even as we added two disks. The number of disks on $A$ is still odd as we removed two disks.
  \end{itemize}
  So, because of $\IH(n-2)$, $\move{010}{n-2}{A}{B}{C}$ is correct.
\end{enumerate}
So $\move{100}{n}{A}{B}{C}$ is correct.

\begin {figure}[h]
\begin{center}
\parbox{5cm}{
% \begin{center}  \input{BouncingTowers/twoDisksOnOddPeg.pstex_t}\end{center}
 \begin{center}  \includegraphics{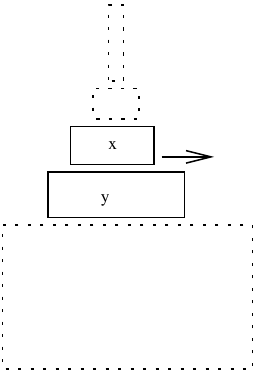}\end{center}
\begin{center}$(i)$\end{center}
  $n$ odd: $a$ is removed first,\\
  $y$ is removed second.
}
\parbox{5cm}{
% \begin{center}  \input{BouncingTowers/twoDisksOnEvenPeg.pstex_t}\end{center} 
  \begin{center}
  \includegraphics{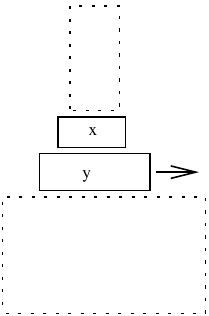}
  \end{center}
\begin{center}$(ii)$\end{center}
  $n$ even: $y$ is removed first,\\
  $x$ is removed second.
}
\end{center}
\caption{Removal order of the last two disks.\label{TwoLastDisksRemoved}}
\end {figure}

\item {Analysis of $\move{001}{n}{A}{B}{C}$:} This function is the exact symmetric of $\move{100}{n}{A}{B}{C}$, for a task exactly symmetric, so has a symmetric proof of its correctness.

\item {Analysis of $\move{010}{2}{A}{B}{C}$:} The two disks (let us call them $\alpha$ and $\beta$) are in position (given fig.  \ref{TwoLastDisksRemoved}, $(ii)$) such that the \texttt{removal} order on $A$ is $(\beta,\alpha)$ and the \texttt{insertion} order on $C$ is $(\alpha,\beta,)$, as $A$ and $C$ have the same parity.  $\beta$ can be inserted on $B$ and they can both be inserted on $C$ because of requirement $(i)$.  So the two disks are correctly moved in three steps, using peg $B$ to dispose temporally disk $\beta$.  So $\move{010}{2}{A}{B}{C}$ is correct.

\item {Analysis of $\move{010}{n}{A}{B}{C}$ if $n>2$:} All along of this proof of correctness we shall use the fact that fixing $2$ disks on the same peg doesn't change the parity of this peg.
\begin{enumerate}
\item $\move{010}{n-2}{A}{B}{C}$ is correct as: from $(i)$ for the initial call results $(i)$ for the first recursive call; $(p)_{010}$ is a natural consequence of $(p)_{010}$ for the initial call (because parity conserved when icing two disks).  So $\IH(n-1)$ implies that $\move{010}{n-2}{A}{B}{C}$ is correct.
\item $A$ and $B$ having different parities, we can move two consecutive disks in two consecutive calls as for $\move{100}{n}{A}{B}{C}$.
\item The second recursive call to $\move{010}{n-2}{A}{B}{C}$ verifies conditions $(i)$ and $(p)_{010}$ as only two extremes disk have been removed from $A$.
\item The two next steps are feasible because of the difference of parity between $B$ and $C$ (same argument as point $2$).
\item The last recursive call is symmetric to the first call, as we move back the $n-2$ disks between the two extreme disk, but this time on $C$.
\end{enumerate}
So $\move{010}{n}{A}{B}{C$} is correct. \qedhere
\end{itemize}
\end{lproof}
\begin{SHORT}

The proof of correctness is merely a study case over the three functions $\move{000}{n}{A}{B}{C}$, $\move{100}{n}{A}{B}{C}$ and $\move{010}{n}{A}{B}{C}$: for lack of space, we defer it to the appendix.
\end{SHORT}
We analyze the complexity of this solution in the next section.

\subsection{Complexity of the algorithm}\label{sec:complexity-algorithm}

Let $f_{xyz}(n)$ be the complexity of the function $\move{xyz}{n}{A}{B}{C}$,
when $|A|\geq n$, $|A|-n \equiv x \mod 2$, $|B| \equiv y \mod 2$ and $|C| \equiv z \mod 2$.
The algorithms from Figures~\ref{fig:move000} to~\ref{fig:move010} yield a recursive system of four equations.\begin{LONG}
\[
\left\{
\begin{array}{lll}
\forall x,y,z & f_{xyz}(0) &= 0\\
\forall x,y,z & f_{xyz}(1) &= 1\\
              & f_{010}(2) &= 3\\
\\
\forall n>1, &f_{000}(n) &= f_{100}(n-1) + 1 + f_{001}(n-1)\\
\forall n>1, &f_{100}(n) &= f_{100}(n-2) + 2 + f_{010}(n-2)\\
\forall n>1, &f_{001}(n) &= f_{010}(n-2) + 2 + f_{001}(n-2)\\
\forall n>2, &f_{010}(n) &= 3 f_{010}(n-2) + 4
\end{array} 
\right.
\]
\end{LONG} As $f_{001}$ is defined exactly as $f_{100}$ (because of the symmetry between $\move{001}{n}{A}{B}{C}$ and $\move{100}{n}{A}{B}{C}$), we can replace each occurence of $f_{001}$ by $f_{100}$, hence reducing the four equations to a system of three equations:
\[
\left\{
\begin{array}{lll}
\forall x,y,z & f_{xyz}(0) &= 0\\
\forall x,y,z & f_{xyz}(1) &= 1\\
              & f_{010}(2) &= 3\\
\\
\forall n>1, &f_{000}(n) &= 2 f_{100}(n-1) + 1 \\
\forall n>1, &f_{100}(n) &= f_{100}(n-2) + 2 + f_{010}(n-2)\\
\forall n>2, &f_{010}(n) &= 3 f_{010}(n-2) + 4
\end{array} 
\right.
\]

\begin{LONG}
\begin{figure}
\centering
$$
\begin{array}{c*{16}{|c}}
 n	   &	 0  &	 1  &	 2  &	 3  &	 4  &	 5  &	 6  &	 7  &	 8  &	 9  &	 10  &	 11  &	 12  &	 13  &	 14  &	 15 \\ \hline
 f_{010}	   &	0  &	1  &	3  &	7  &	13  &	25  &	43  &	79  &	133  &	241  &	403  &	727  &	1213  &	2185  &	3643  &	6559 \\ \hline
 f_{100}	   &	0  &	1  &	2  &	4  &	7  &	13  &	22  &	40  &	67  &	121  &	202  &	364  &	607  &	1093  &	1822  &	3280 \\ \hline
 f_{000}	   &	0  &	1  &	3  &	5  &	9  &	15  &	27  &	45  &	81  &	135  &	243  &	405  &	729  &	1215  &	2187  &	3645 \\ \hline
 3^{\lceil n/2\rceil }	   &	 1  &	 3  &	 3  &	 9  &	 9  &	 27  &	 27  &	 81  &	 81  &	 243  &	 243  &	 729  &	 729  &	 2187  &	 2187  &	 6561 \\
\end{array}
$$
\caption{The first values of $f_{010}$,$f_{100}$ and $f_{000}$, computed automatically from the recursion. those corrobolate the intuition that $f_{100}(n)<f_{000}(n)$ for values of $n$ larger than $1$.}
\label{fig:firstValues}
\end{figure}
\end{LONG}

Lemmas~\ref{res:f010} to \ref{res:f000} resolve the system function by function.  
The function $f_{010}(n)$ can be solved independently from the others:

\begin{lemma}\label{res:f010}
$f_{010}(n)=  \left\{ \begin{array}{ll}
                         0                              & \mbox{if $n=0$;}                 \\ 
                         1                              & \mbox{if $n=1$;}                 \\ 
                         3                              & \mbox{if $n=2$;}                 \\ 
                         3^{\frac{n+1}{2}} - 2 & \mbox{if $n\geq 3$ is odd; and } \\
                         5 \times 3^{\frac{n}{2}-1} - 2 & \mbox{if $n\geq 4$ is even.} 
                      \end{array} \right.
$
\end{lemma}
\begin{proof}
Consider the recurrence $X_{k+1}=3X_k +4$ at the core of the definition of $f_{010}$: a mere extension yields the simple expression $X_k=3^k(X_0+2)-2$.
\begin{itemize}
\item When $n\geq 3$ is odd, set $k=\frac{n-1}{2}\geq 1$,  $U_0=1$ and $U_{k+1}=3U_k +4$ so that $f(2k+1)=U_k=3^k(1+2)-2$.  Then $f_{010}(n)=3\times 3^{k} -2=3^{k+1} -2$ for $n\geq 3$ and odd.
\item When $n\geq 4$ is even, set $k=\frac{n}{2}\geq 1$, $V_0=3$ and  $V_{k+1}=3V_k +4$  so that $f(2k)=V_k=3^k(3+2)-2$, so that $f_{010}(n)=5\times 3^k -2$ for $n\geq 4$ and even.
\end{itemize}
Gathering all the results yields the final expression.
\end{proof}

The expression for the function $f_{010}$ yields the expression for the function $f_{100}$:
\begin{lemma} \label{res:f100}
$f_{100}(n) = \left\{ 
\begin{array}{ll}
  0                                 & \mbox{ if $n=0$;}                   \\
  1                                 & \mbox{ if $n=1$;}                   \\
  2                                 & \mbox{ if $n=2$;}                   \\
  4                                 & \mbox{ if $n=3$;}                   \\
  \frac{5}{2}\times 3^{\frac{n}{2}-1} +2 & \mbox{ where $n\geq4$ is even; and }  \\
  \frac{3^{\frac{n+1}{2}}-1}{2} & \mbox{ where $n\geq5$ is odd. } \\
\end{array}
\right.
$ 
\end{lemma}
\begin{proof}
Consider the projection of the system to just $f_{100}$:
\[ f_{100}(n) = \left\{ 
\begin{array}{ll}
  0 & \mbox{ if $n=0$}\\
  1 & \mbox{ if $n=1$}\\
  f_{100}(n-2) +2 + f_{010}(n-2) & \mbox{ if $n\geq2$ }\\
\end{array}
\right.
\]

For any integer value of $k\geq0$, we combine some change of variables with the results from Lemma~\ref{res:f010} to yied two linear systems, which we solve separately:
\begin{itemize}
\item $V_k = f_{100}(2k)$ and $V_0=f_{100}(0)=0$ so that $f_{100}(n) = V_{k}$ if $n$ is even and $k=\frac{n}{2}$; and 
\item $U_k = f_{100}(2k+1)$ and $U_0=f_{100}(1)=1$ so that $f_{100}(n) = U_{k}$ if $n$ is odd and $k=\frac{n-1}{2}$.
\end{itemize}

On one hand, $U_k = U_{k-1} + 2 + f_{010}(2k+1-2)$ for $k>0$ and $U_0=1$. 
This yields a linear recurrence which we develop as follow:
\begin{eqnarray*}
U_k & = & U_{k-1} + 2 + f_{010}(2k-1)   \mbox{ by definition;}                                                      \\ 
    & = & U_{k-1} + 2 + 3\times 3^{\frac{(2k-1)-1}{2}} - 2  \mbox{ via Lemma~\ref{res:f010} because $2k-1$ is odd;} \\
    & = & U_{k-1} + 3^k               \mbox{ by mere simplification;}                                               \\
    & = & U_0 + \frac{3}{2}(3^k-1)               \mbox{ by resolution of a geometric serie;}                        \\
    & = & \frac{3^{k+1}-1}{2} \mbox{ because $U_0=1$.}  
\end{eqnarray*}
Since $f_{100}(n) = U_{\frac{n-1}{2}}$ when $n$ is odd, the solution above yields $f_{100}(n) = \frac{3^{\frac{n+1}{2}}-1}{2}$ if $n$ is odd.

On the other hand, $V_k=V_{k-1} + 2 + f_{010}(2k-2)$ for $k>0$ and  $V_0=0$. 
The initial conditions of $f_{010}$ for $n=0,1$ and $2$ yields the three first values of $V_k$: 
$V_0=0$; 
$V_1= V_0 + 2 + f_{010}(0) = 0+2+0 = 2$; and 
$V_2= V_1 + 2 + f_{010}(2) = 2+2+3=7$. 
Then we develop the recursion for $k\geq 3$ similarly to $U_k$:
\begin{eqnarray*}
  V_k & = & V_{k-1} + 2 + f_{010}(2k-2)   \mbox{ by definition;}                                                                         \\ 
      & = & V_{k-1} + 2 + 5\times 3^{\frac{(2k-2)}{2}-1} - 2  \mbox{ for $2k-2\geq 4$ even, or any  $k\geq 3$ via Lemma~\ref{res:f010};} \\
      & = & V_{k-1} + 5\times 3^{k-2}   \mbox{ by mere simplification (still only for $k\geq 3$);}                                    \\
      & = & V_2 + 5 \left( 3^1+\cdots + 3^{k-2}\right)        \mbox{ by propagation;}                                                         \\
      & = & V_2 + 5 \frac{ 3^{k-1}-2 }{2}                      \mbox{ by resolution of a geometric serie;}                                 \\
      & = & 7 + \frac{5}{2} (3^{k-1}-2)   \mbox{ because $V_2=7$;}  \\
      & = & \frac{5}{2} 3^{k-1} +2   \mbox{ by simplification.}  
\end{eqnarray*}
Since $f_{100}(n) = V_{\frac{n}{2}}$ when $n$ is even, the solution above yields $f_{100}(n) = \frac{5}{2} 3^{\frac{n}{2}-1} +2$ if $n$ is even.

Reporting those results in the definition of $f_{100}$ yields the final formula\begin{LONG}:
$$f_{100}(n) = \left\{ 
\begin{array}{ll}
  0                                 & \mbox{ if $n=0$;}                   \\
  1                                 & \mbox{ if $n=1$;}                   \\
  2                                 & \mbox{ if $n=2$;}                   \\
  4                                 & \mbox{ if $n=3$;}                   \\
  \frac{5}{2}\times 3^{\frac{n}{2}-1} +2 & \mbox{ where $n\geq4$ is even; and }  \\
  \frac{3^{\frac{n+1}{2}}-1}{2} & \mbox{ where $n\geq5$ is odd. } \\
\end{array}
\right.
$$ 
\end{LONG}\begin{SHORT}.\end{SHORT}
\end{proof}

Finally, the expression for the function $f_{100}$ directy yields the expression for the function $f_{000}$:

\begin{lemma} \label{res:f000}
$f_{000}(n) = \left\{ 
\begin{array}{ll}
  1 & \mbox{ if $n=1$}\\
  3 & \mbox{ if $n=2$}\\
  5 & \mbox{ if $n=3$}\\
  3^{\frac{n}{2}} & \mbox{ where $n\geq4$ is even; and } \\
  5(3^{\frac{n-3}{2}} + 1) & \mbox{ where $n\geq5$ is odd.}  \\
\end{array}
\right.
$ 
\end{lemma}

\begin{proof}
\[f_{100}(n) = \left\{ 
\begin{array}{ll}
  1                                 & \mbox{ if $n=1$;}                   \\
  2                                 & \mbox{ if $n=2$;}                   \\
  4                                 & \mbox{ if $n=3$;}                   \\
  \frac{5}{2} 3^{\frac{n}{2}-1} +2 & \mbox{ where $n\geq4$ is even; and }  \\
  \frac{3^{\frac{n+1}{2}}-1}{2} & \mbox{ where $n\geq5$ is odd. } \\
\end{array}
\right.
\] 

From these results, deduce the value of $f_{000}(n)$ using that $f_{000}(n) = 2 f_{100}(n-1) + 1$.

\[ f_{000}(n) = \left\{ 
\begin{array}{ll}
  1 & \mbox{ if $n=1$}\\
  3 & \mbox{ if $n=2$}\\
  5 & \mbox{ if $n=3$}\\
  5\times 3^{\frac{n-1}{2}-1} +5 & \mbox{ where $n\geq5$ is odd; and }  \\
  3^{\frac{n}{2}} & \mbox{ where $n\geq6$ is even. } \\
\end{array}
\right.
\]
\end{proof}

% Note that as $\frac{5}{3}\sqrt{3}^{n-1}-\frac{2}{3}=3^{\frac{n-1}{2}}+2(3^{\frac{n-3}{2}}-1)$, $f_{000}(n)$ is an integer even if $n$ is odd.
%As $\frac{5}{3}\approx 1.67 < \sqrt{3}\approx1.73$, even when $n$ is odd, the value of $f_{000}(n) = \frac{5}{3}\sqrt{3}^{n-1}-\frac{2}{3}$ is less than $\sqrt{3}^n$.  
As $\sqrt{3}\approx1.73<2$, this value is smaller than the number $2^n-1$ of steps required to move a \hanoit.  We prove that this is optimal in the next section.

\section{Optimality} \label{sec:optimality}

Each legal state of the \bouncingtpb\ with three pegs and $n$ disks can be uniquely described by a word of length $n$ on the three letters alphabet $\{A,B,C\}$, where the $i$-th letter indicates on which peg the $i$-th largest disk stands.  Moreover, each word of $\{A,B,C\}^n$ corresponds to a legal state of the tower, so there are $3^n$ different legal states (even though not all of them are reachable from the initial state). 

To prove the optimality of our algorithm, we prove that it moves the disks along the shortest path in the {\em configuration graph} (defined in Section~\ref{sec:configuration-graph}) by a simple induction proof (in Section~\ref{sec:proof-optimality}).

\subsection{The configuration graph}\label{sec:configuration-graph}

The configuration graph of a \bouncingt\ has $3^n$ vertices corresponding to the $3^n$ legal states, and two states $s$ and $t$ are connected by an edge if there is a legal move from state $s$ to state $t$. The reversibility of moves (seen in Section~\ref{sec:struct-facts-single}) implies that the graph is undirected.

Consider the initial state $A\ldots A$ ($=A^n$).  The smallest disk $1$ cannot be moved before the other disks are all moved to peg $B$ or all moved to peg $C$: we can't remove disk $1$ from peg $A$ if there is a disk under it, and we can't put it on another peg if a larger disk is already there.  This partitions $G$ into three parts, each part being characterized by the position of disk $1$; these parts are connected by edges representing a move of disk $1$ (see the recursive decomposition of $G(n)$ in Figure~\ref{TotalGraphForHanoi}).

Each part is an instance of the configuration graph $G'(n-1)$ defining all legal steps of $(n-1)$ disks $\{2,\ldots,n\}$ given that disk $1$ is fixed on its peg.

\begin {figure}[h]
\centering
\includegraphics[width=.8\textwidth]{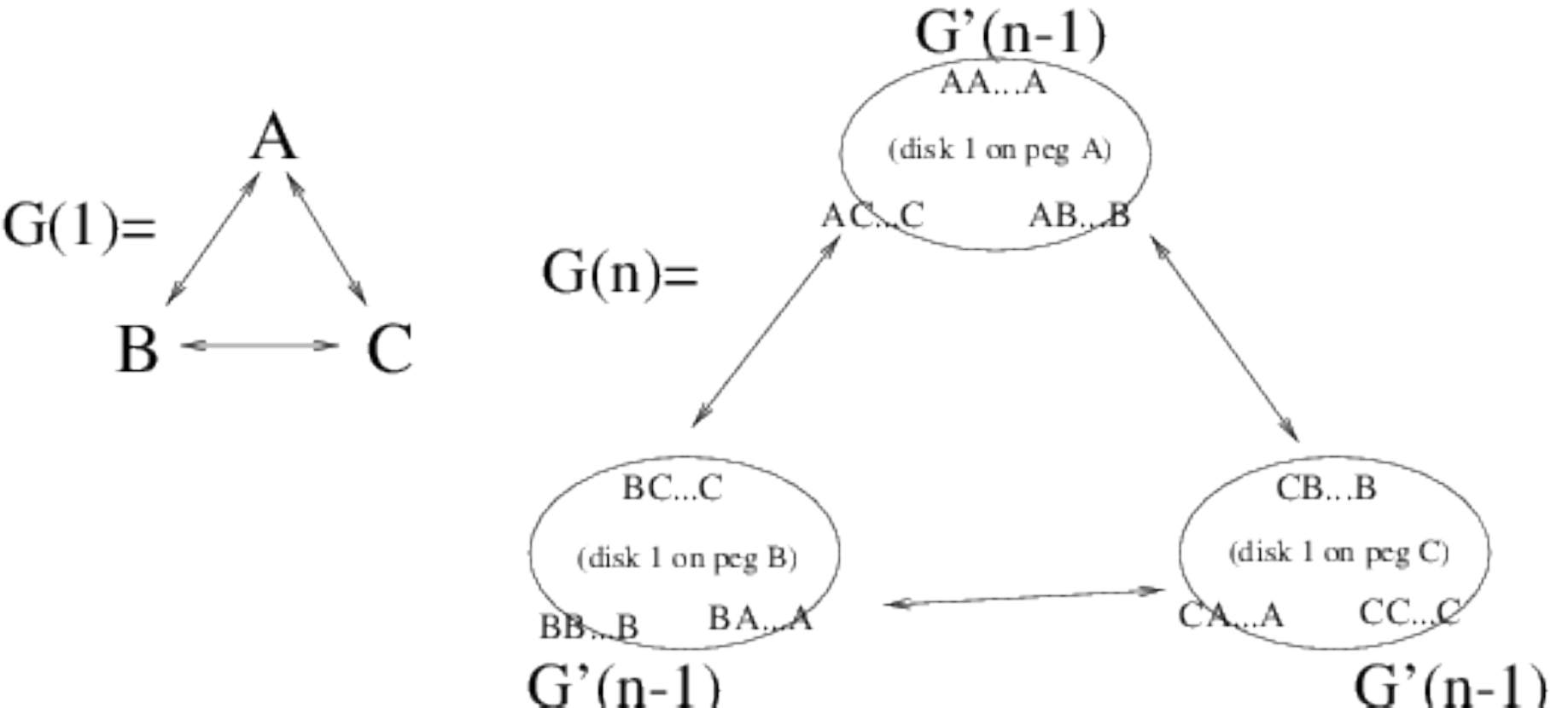}
\caption{First decomposition of the configuration graph of the \bouncingtpb.  \label{TotalGraphForHanoi}}
%   \parbox{6cm}{
%     \includegraphics[width=6cm]{subdivisionOfTotalGraphForHanoi}
%     \caption{Subdivision of the configuration graph.
%       \label{subdivisionOfTotalGraphForHanoi}
%       }
%     }
%   \hfil
%  \parbox{7cm}{
%\input{BouncingTowers/totalGraphForHanoi.pstex_t}
\end{figure}

% Graphics Macros.
\newcommand\f{\rightarrow}
\newcommand\A{{{a}}}
\newcommand\B{{{b}}}
\newcommand\C{{{c}}}

Let us consider this subgraph $G'(n-1)$, when disk $1$ (the smallest) is fixed on one peg (say on peg $A$).  Note each state of this graph $\A X\ldots Z$, where $\A$ stands for the disk $1$ fixed on peg $A$, and $X\ldots Z$ for positions of other disks on diverse pegs.  The removal order changes from those observed in $G$ each time $|A|$ is odd.
%(see section \ref{removalOrderChangeWithParity})

To remove the two extreme disks $2$ and $n$ (not moving disk $1$, since it is fixed), it is necessary to move all other disks to a single other peg (same argument as for $G(n)$), so we can divide our configuration graph in subsets of states corresponding to different positions where disks $2$ and $n$ are fixed.
%parts corresponding to the different positions where icing these two disks.

This defines $9$ parts, as each of the two fixed disks can be on one of the three peg.
Of those $9$ parts, we need focusing only on $5$:
\begin{itemize}
\item two parts of the graph cannot be accessed from the initial state $\A A\ldots A$, (see an illustration in Figure~\ref{somePartsCannotBeAccessed}); and
\item the part of the graph where disk $2$ is fixed on $B$ and disk $n$ is fixed on $C$ contains two parts, which are not connected for $n>4$ (see an illustration in Figure~\ref{somePartCanBeUnConnex}).
  \end{itemize}

\begin{figure} \centering
\parbox{.2\textwidth}{
\hanoiState
{\peg{\deux}{\unIced}{}{}}
{\peg{\n}{}{}{}}
{\peg{\nmun}{\diskDots}{\trois}{}}
}
\parbox{.2\textwidth}{
\hanoiState
{\peg{\deux}{\unIced}{}{}}
{\peg{\nmun}{\diskDots}{\trois}{}}
{\peg{\n}{}{}{}}
} 
\parbox{.5\textwidth}{ \caption{States where disk $2$ is on $A$ and disk $n$ is on another peg (i.e. $B$ or $C$) cannot be accessed from the initial state $A\ldots A$ for $n>4$.  No move is possible from these states as $A$ cannot receive larger disk than $2$ (and all are), $B$ cannot receive smaller disk than $n$ (and all are), and $C$ cannot receive disk $2$ nor $n$ if $n>4$.
 \label{somePartsCannotBeAccessed}}}
%%  from the initial state for $n>4$
%%  because:
%%  \begin{itemize}
%%  \item To access $2$ or $n$ all other disks have to be removed to another peg;
%%  \item To move $n$, $2$ must be removed and one peg must be freed, which
%%  can't happen as one peg is taken by the inner disks $3\ldots n-1$,
%%  another one by $2$ and the third one by $1$ and $n$ itself.
%%  \end{itemize}
 \end{figure}

\begin{figure} \centering
\parbox{.2\textwidth}{ \hanoiState { \peg{\nmun}{\diskDots}{\trois}{\unIced} } {
    \peg{\deux}{}{}{} } { \peg{\n}{}{}{} } }
\parbox{.2\textwidth}{ \hanoiState { \peg{\unIced}{}{}{} } {
    \peg{\nmun}{\diskDots}{\trois}{\deux} } { \peg{\n}{}{}{} } }
\parbox{.5\textwidth}{ \caption{States $\A B A\ldots A C$ and $\A B B\ldots B C$ are not connected in the subgraph where disks $2$ and $n$ are fixed on $B$ and $C$, and disk $1$ is fixed on $A$: As no disk can be inserted under $n$, if $n>4$ it is impossible to move the $n-3>1$ unfixed disks from $A$ to $B$ (as to move more than one disk between two pegs of same parity require a third peg).\label{somePartCanBeUnConnex} 
% No disk can be inserted below disk $n$, and so, if
%     $n>4$, it is impossible to move the $n-3>1$ unfixed disks from
%     $A$ to $B$ (resp. $C$) and vice-versa, as to move more than one
%     disk between two pegs of same parity need a third peg.  
}}
\end{figure}

\begin{INUTILE}
\begin {figure}[h]
\parbox{5cm}{
  \begin{center}
    \includegraphics[width=5cm]{totalGraphDecompositionWithOneDiskIced}
  \end{center}
}
\parbox{10cm}{
  \caption{Decomposition of the configuration graph with one disk fixed.}
  \label{TotalGraphDecompositionWithOneDiskIced}
}
\end {figure}
\end{INUTILE}

The five remaining parts are very similar.  Three of them are of particular importance as each contains one key state, which are $\A A\ldots A$, $\A B\ldots B$ and $\A C\ldots C$.  Consider first the graphs $G'(n)$ for $n\in\{1,2,3\}$ ($n+1$ disks in total if we count the fixed one): they are represented in Figure~\ref{smallGraphs}.  When one disk is fixed on $A$, the task of moving disks from $A$ to $B$ is symmetric with moving them from $A$ to $C$, but quite distinct from the task of moving disks from $B$ to $C$.

\begin {figure}\centering
  \includegraphics{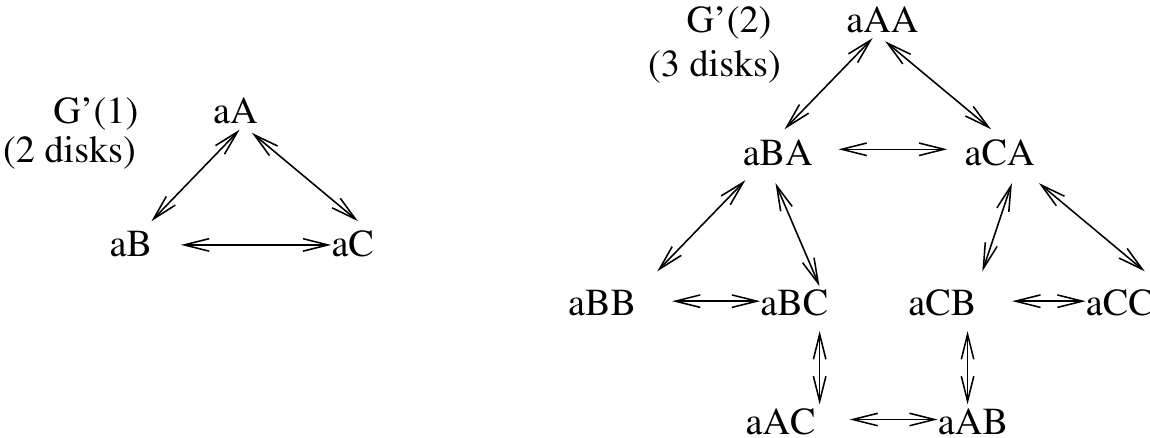}
  \includegraphics[width=\textwidth]{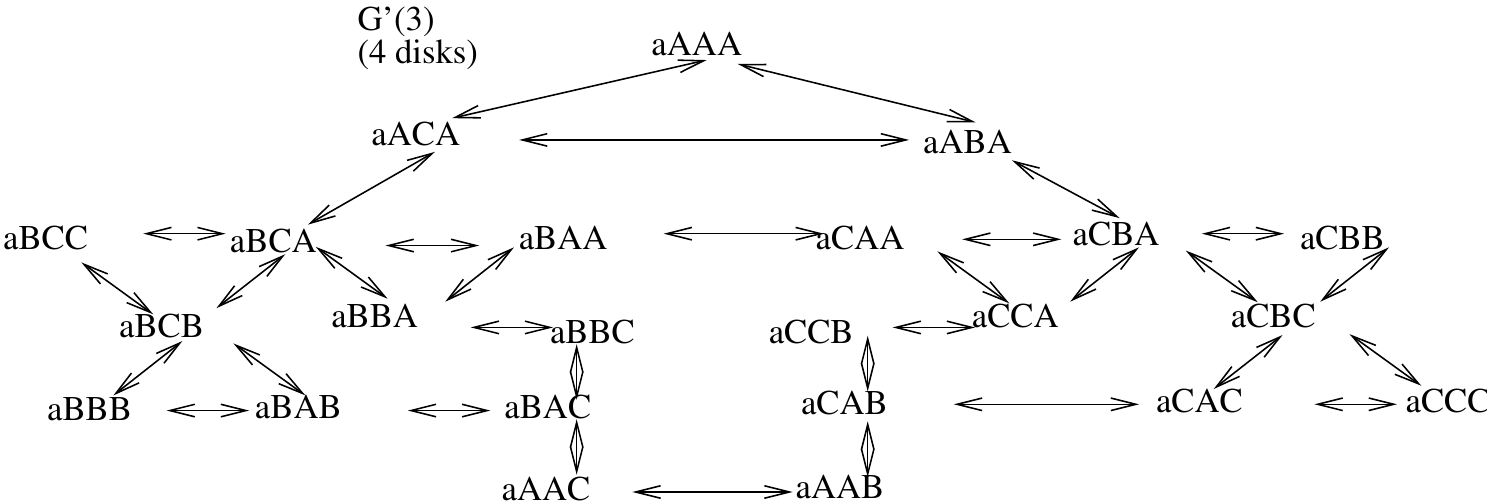}
  \caption{Subgraphs $G'(n)$ with one disk fixed on the peg $\pegA$ for $n\in\{1,2,3\}$.\label{smallGraphs}}
\end {figure}

% \begin {figure}[h]
%  \input{BouncingTowers/totalGraph3DisksAndOneIced.pstex_t}
%   \caption{Subgraphs with one disk fixed among $4$ disks.}
%   \label{smallGraphs}
% \end {figure}

%\input{totalGrapheForN=1.pscript}
%\input{totalGrapheForN=2.pscript}
%\input{totalGrapheForN=3.pscript}
%\input{totalGrapheForN=4.pscript}
%\input{totalGrapheForN=5.pscript}
%\input{totalGrapheForN=6.pscript}
%\input{totalGrapheForN=7.pscript}

Now, consider the part of the graph $G'(n-1)$ where the smallest and the largest disks ($2$ and $n$) are fixed on $A$.  This part contains the initial state $A\ldots A$.  The only way to free the smallest disk is to move the $n-3$ other disks to another peg.
\begin{TODO}
Male another figure  illustrating how ``The only way to free the smallest disk is to move the $n-3$ other disks to another peg.''
 (see an illustration in Figure~\ref{totalGraphDecompositionWithOneDiskIcedA})
\end{TODO}

Once disks $2$ and $n$ are fixed on the same peg (in addition to disk $1$), the situation is similar to the entire graph, with two fewer disks.  It is the case each time two extreme disks are fixed on the same peg: when $2$ and $n$ are fixed on peg $C$ or $B$, or when $1$ and $n$ are fixed on peg $A$; the process can then ignore the two fixed disks to move the $n-3$ remaining disks, as the parity of the peg is unchanged.  See the definitions of the graph $G'(n)$ in Figure~\ref{smallGraphs} for $n\in\{1,2,3\}$ and in Figure~\ref{totalGraphDecompositionWithOneDiskIcedA} for $n>3$.

% \begin {figure}[h]
% \center{  
% \includegraphics[width=7cm]{recursiveDefinitionOfTheGraphWithOneDiskIced}}
% %    \includegraphics[width=7cm]{totalGraphDecompositionWithOneDiskIced}
% %totalGraphWithOneDiskIced}
%   \caption{Recursive definition of the configuration graph with one disk fixed.}
%   \label{TotalGraphWithOneDiskIced}
% \end {figure}

\begin{figure}[ht]
\begin{center}
\includegraphics[width=\textwidth]{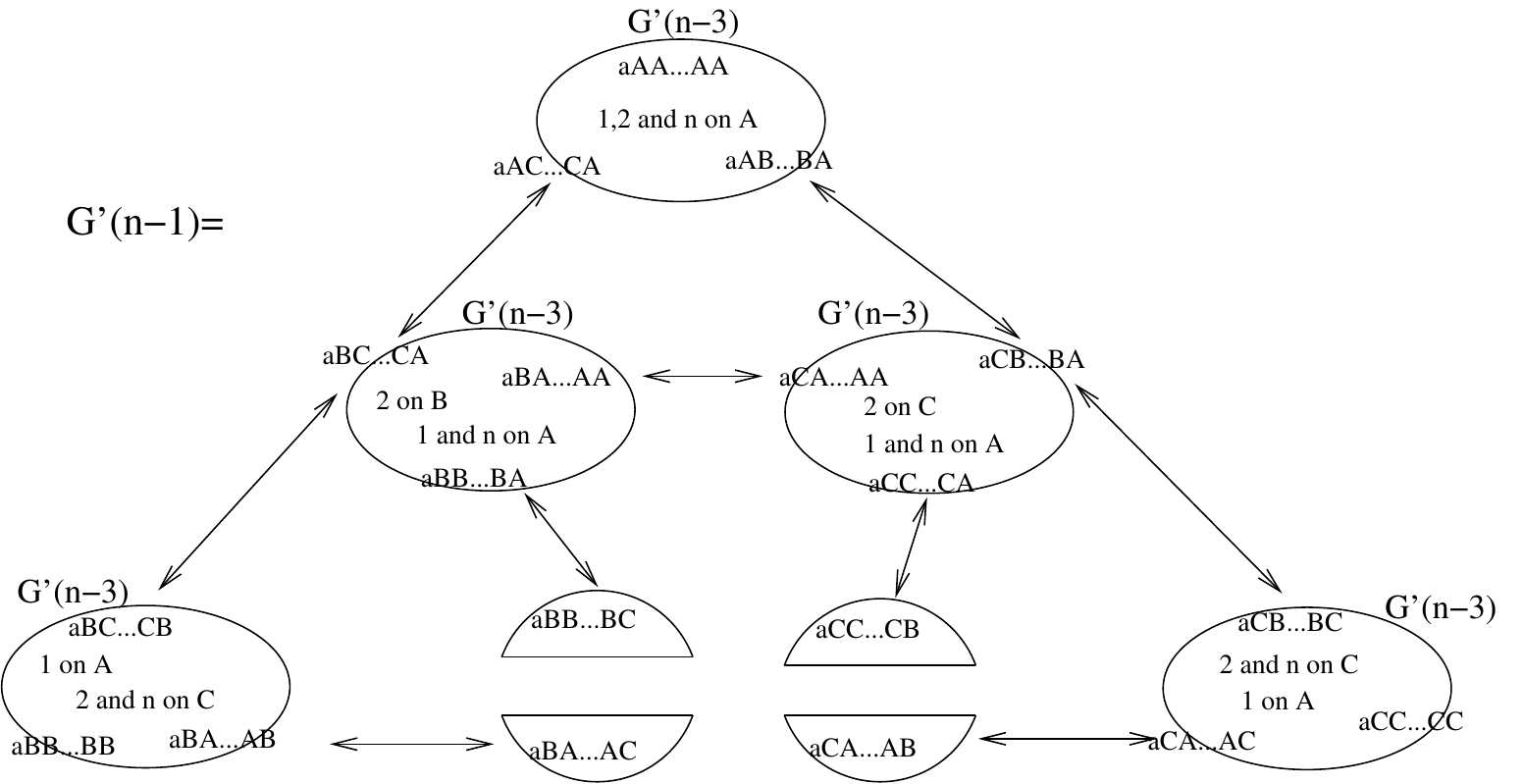}
\end{center}
\caption{Recursive definition of $G'(n)$, the graph of all legal steps when one disk is fixed on the first peg,~for $n>3$. There is no way to connect the states $aBB....BC$, $aBA...AC$, $aCC...CB$ and $aCA...AB$ without moving some of the disks from $\{1,2,n\}$.
\label{totalGraphDecompositionWithOneDiskIcedA}}
\end{figure}

\subsection{Proof of optimality}\label{sec:proof-optimality}

To prove the optimality of the solution described in Section~\ref{sec:solution}, we prove that the algorithm is taking the shortest path in the configuration graph defined in the last section.  A side result is that this is the unique shortest solution.
\begin{SHORT}
Once understood the definition of the configuration graph, the proof of optimality is merely a study case: for lack of space, we defer it to the appendix.
\end{SHORT}

\begin{INUTILE}
\begin{lemma}
From one state in the configuration graph, we can have at most five edges.
\end{lemma}
\begin{lproof}
Three disk can be moved - at most one per peg - and each one can be moved to two peg at most {\em but} states where a disk $\alpha$ can be moved from $A$ to $B$ and another disk $\beta$ can be moved form $B$ to $A$ are subject to constraint $A$ and $B$ mustn't have the same parity.  So we can have at most $2$ two-way steps (that's a coloration problem), which means four steps, plus $1$ one way move.  (For instance $ACBA$ is a state connected with states $AABA$, $ACBB$, $ACBC$, $ACAA$,$ACCA$)
\end{lproof}
\end{INUTILE}

\begin{theorem} \label{res:proof-optimality}
$\forall(x,y,z)\in\{0,1\}^3,\, \forall n\geq 0,\, \move{xyz}{n}{A}{B}{C}$ moves optimally $n$ disks from $A$ to $C$.
\end{theorem}
\begin{lproof}

\begin{TODO}
CHANGE definition of Induction hypothesis to have the forall inside.
\end{TODO}
Define the induction hypothesis $\IH(n)$ as ``$\forall(x,y,z)\in\{0,1\}^3\, \move{xyz}{n}{A}{B}{C}$ moves optimally $n$ disks from $A$ to $C$''. Trivially $\IH(0)$ and $\IH(1)$ are true.  Suppose that there exists an integer $N>1$ such that $\forall n<N$, the induction hypothesis $\IH(n)$ is true.  We prove that $\IH(N)$ is then also true.

\begin{itemize}

\item $\move{000}{N}{A}{B}{C}$ is optimal:

%(please see Figure~\ref{subdivisionOfTotalGraphForHanoi})
$\move{000}{N}{A}{B}{C}$ for $N>0$ consists of one
call to $\move{100}{N}{A}{C}{B}$, one unitary step,
and one call to $\move{001}{N}{B}{A}{C}$.

So it moves optimally (by $\IH(N-1)$) from $\A A\ldots A$ to $\A
B\ldots B$, and then to $\C B\ldots B$, and after that to $\C C\ldots C$.
(In Figure~\ref{TotalGraphForHanoi} the right edge of
the triangle.)

A path not going through states $\A B\ldots B$ or $\C B\ldots B$ would take more steps:
\begin{itemize}
\item if we don't go through the state $\A B\ldots B$, then the state $\A C\ldots C$ is necessary, with a cost of $f_{100}(N-1)$, and also the state $\B C\ldots C$ (with a cost of $1$), and at the end of the path we have to go through the state $\C A\ldots A$, which optimal path to go to the final $\C C\ldots C$ state is of length $f_{100}(N-1)$: this path is of length $f_{100}(N-1)+1+f_{100}(N-1)$ and is already as long as the one given by $\move{000}{N}{A}{B}{C}$.
\item if we go through $\A B\ldots B$, but not through $\C B\ldots B$, then the path is not optimal as it must go through $\A C\ldots C$ and the optimal path from $\A A\ldots A$ to $\A C\ldots C$ doesn't go through $\A B\ldots B$.
\end{itemize}
So $\move{000}{N}{A}{B}{C}$ is optimal.

\item $\move{100}{N}{A}{B}{C}$ is optimal:

  $\move{100}{N}{A}{B}{C}$ for $N>1$ consists of one call to $\move{100}{N-2}{A}{C}{B})$, two steps, and one call to $\move{010}{N-2}{B}{A}{C})$.

As before, we shall consider these recursive calls of order smaller
than $N$ as optimal because of $\IH(N-2)$. So we know how to move
optimally 
from    $\A A A\dots A A$ 
to      $\A A B\dots B A$, 
to      $\A C B\dots B A$, 
then to $\A C B\dots B C$ 
and to  $\A C C\dots C C$
(in figure \ref{totalGraphDecompositionWithOneDiskIcedA}), this
corresponds to the left edge of the triangle).

We must now prove that other paths take more steps:
\begin{itemize}
\item We cannot avoid the state $\A C B\dots B C$, neither $\A C
  B\dots B A$, as there is no other way out of $\A C C\dots C C$.
\item if we avoid the state $\A A B\dots B A$ then the optimal path to
  $\A C B\dots B A$ necessarily passes by $\A BA\ldots AA$ and $\A
  CA\ldots AA$, and is of length
  $f_{100}(N-2)+1+f_{010}(N-2)+1+f_{010}(N-2)$, which is longer than
  the whole solution given by the algorithm, of length $f_{100}(N)=
  f_{100}(N-2) + 2 + f_{010}(N-2)$.
\end{itemize}
So $\move{100}{N}{A}{B}{C}$ is optimal.

\item $\move{010}{N}{A}{B}{C}$ is optimal:

$\move{010}{1}{A}{B}{C}$ and $\move{010}{2}{A}{B}{C}$ are special cases, 
we can see in graphs $G'(1)$ and $G'(2)$ on figure 
\ref{smallGraphs} page \pageref{smallGraphs} 
that the optimal paths between $\A B\ldots B$ and $\A C\ldots C$
are of length $1$ and $3$, as the solutions produced by the algorithm.
So $\move{010}{1}{A}{B}{C}$ and $\move{010}{2}{A}{B}{C}$ are proven optimal.

$\move{010}{N}{A}{B}{C}$ for $N>2$ corresponds to a
path going through the states (the first disk being fixed on $\B$):
(please report to fig. \ref{totalGraphDecompositionWithOneDiskIcedA} 
from $\A C C\ldots CC$ to $\A B B\ldots BB$ down left to down right. )
 \[
 \A C C\dots C C
\stackrel{f_{010}(N-2)}{\longrightarrow}
 \A C B\dots B C
\stackrel{2}{\longrightarrow}
 \A A B\dots B C
\]\[
\stackrel{f_{010}(N-2)}{\longrightarrow}
 \A A C\dots C C 
\stackrel{2}{\longrightarrow}
 \A B C\dots C B
\stackrel{f_{010}(N-2)}{\longrightarrow}
 \A B B\dots B B 
\]

% \begin{figure}[h]
% \parbox{5cm}{
% \[\begin{array}{ll}
% \B A A\dots A A\\
% &\downarrow f_{010}(N-2)\\
% \B A C\dots C A\\
% &\downarrow 2\\
% \B B C\dots C A\\
% &\downarrow f_{010}(N-2)\\
% \B B A\dots A A\\
% &\downarrow 2\\
% \B C A\dots A C\\
% &\downarrow f_{010}(N-2)\\
% \B C C\dots C C
% \end{array}\]
% }
% \parbox{8cm}{
% \input{BouncingTowers/totalGraphEtiquettedWithOneDiskFixedB.pstex_t}
% }
% \caption{Total Graph Decomposition with one fixed disk.
% \label{totalGraphDecompositionWithOneDiskIcedB}}
% \end{figure}

We shall demonstrate that all other paths take more steps:
\begin{itemize}
\item The states $\B A C\dots C A$ and  $\B C A\dots A C$ are mandatory,
  for connexity, and so are  $\B A C\dots C B$ and  $\B C A\dots A B$.
\item if we go through $\B B C\dots C B$, 
then it's $\B B A\dots A B$ which is mandatory.
\item if we contourn $\B B C\dots C B$, then we shall go through 
$\B A B\dots B B$, $\B C B\dots B B$ and $\B C A\dots A B$:
the total path would be of length $3+4 f_{010}(N-2)$,
to be compared with $4 +3 f_{010}(N-2)$ (We trade one step
with one recursive call).
As $f_{010}(N-2) \geq 1$ for $N-2\geq q$ (i.e. $N\geq 3>2$),
$\move{010}{N}{A}{B}{C}$ is optimal for $N>2$.
\end{itemize}
So $\move{010}{N}{A}{B}{C}$ is optimal. \qedhere
\end{itemize}
\end{lproof}

We discuss further extensions of those results in the next section.

\section{Discussion}\label{sec:discussion}

All the usual research questions and extensions about the \hanoitpb\ are still valid about the \bouncingtpb. We discuss only a selection of them, such as the space complexity in Section~\ref{sec:spaceComplexity}, and the extension to other proportional \texttt{insertion} and \texttt{removal} points in Section~\ref{sec:levitating-towers}.

\subsection{Space Complexity}
\label{sec:spaceComplexity}

Allouche and Dress~\cite{1990-RAIRO-ToursDeHanoiEtAutomates-AlloucheDress} showed that the optimal sequence of steps required to move a \hanoit\ of $n$ disks can be obtained by a simple function from the prefix of an infinite unique sequence, which itself can be produced by a finite automaton. This proves that the space complexity of the \hanoitpb\ is constant. 

The same technique does not seem to yield constant space for \bouncingt s: whereas the sequences of steps generated by each of the functions $\move{100}{n}{A}{B}{C}$, $\move{010}{n}{A}{B}{C}$ and $\move{001}{n}{A}{B}{C}$ are prefixes of infinite sequences, extracting those suffixes and combining them in a sequence corresponding to $\move{000}{n}{A}{B}{C}$ would require a counter using logarithmic space in the length of the sequences to be extracted, i.e. $\log_2 (\sqrt{3}^n)\in \Theta(n)$, which would still be linear in the number of disks.

\begin{DISKPILEPROBLEM}
\begin{TODO}
\subsection{Bouncing Disk Piles}
\label{sec:bouncingDiskPiles}

Whereas we studied the \bouncingtpb\ in the case where all the disks were of distinct size, an extension would be to consider the case where not all sizes are distinct.

\end{TODO}
\end{DISKPILEPROBLEM}

\subsection{Levitating Towers}
\label{sec:levitating-towers}\label{alphaltIntroduced}

An extension of the \bouncingtpb\ is to parametrize the \texttt{insertion} point, so that the \texttt{removal} point is at position $\extractingpoint{n}$ and the \texttt{insertion} point is under the disk at position $\insertingpoint{n}$ in a tower of $n$ disks, for $\alpha\in[0,\frac{1}{2}]$ fixed (the problem is symmetrical for $\alpha\in[\frac{1}{2},1]$).  By analogy with \bouncingt s, we call this variant a \alphalt.  This parametrization creates a continuous range of variants, of which the \hanoitpb\ and the \bouncingtpb\ are the two extremes:
\begin{itemize}
\item for $\alpha=0$, the \texttt{removal}/\texttt{insertion} point is always at the top, which corresponds to a \hanoit, while
\item for $\alpha=\frac{1}{2}$ the problem corresponds to a \bouncingt.
\end{itemize}

The complexity of moving a \alphalt\ cannot be smaller than the one of a \bouncingt, as the key configuration permitting to move $2$ disks in $2$ steps between the same pegs is less often obtainable in a \alphalt.

\bibliography{/home/jbarbay/EverGoing/WebSite/Studies/ScientificArticles/biblio-Barbay,/home/jbarbay/EverGoing/WebSite/Publications/publications-Barbay}

\begin{thebibliography}{1}

\bibitem{1990-RAIRO-ToursDeHanoiEtAutomates-AlloucheDress}
J.-P. Allouche and F.~Dress.
\newblock Tours de {Hano\"\i} et automates.
\newblock {\em RAIRO, Informatique Th\'eorique et applications}, 24(1):1--15,
  1990.

\bibitem{1981-IPL-TheCyclicTowersOfHanoi-Atkison}
M.D. Atkinson.
\newblock The cyclic towers of {Hano\"{\i}}.
\newblock {\em Information Processing Letters (IPL)}, 13(118-119), 1981.

\bibitem{1892-BOOK-MathematicalRecreationsAndEssays-Ball}
W.~R. Ball.
\newblock {\em Mathematical Recreations and Essays}.
\newblock McMillan, London, 1892.

\bibitem{1941-AmericanMathematics-SolutionOfProblemNo2918-FrameStewart}
J.~S. Frame and B.~M. Stewart.
\newblock Solution of problem no 3918.
\newblock {\em American Mathematics Monthly (AMM)}, 48:216--219, 1941.

\bibitem{1883-MISC-LaTourDHanoi-Lucas}
\'Edouard Lucas.
\newblock La tour d'{Hano\"i}, v\'eritable casse-t\^ete annamite.
\newblock In a puzzle game., Amiens, 1883.
\newblock Jeu rapport\'e du Tonkin par le professeur N.Claus (De Siam).

\bibitem{1883-BOOK-RecreationsMathematiques-Lucas}
\'Edouard Lucas.
\newblock {\em R\'ecr\'eations Math\'ematiques}, volume~II.
\newblock Gauthers-Villars, Paris, quai des Augustins, 55, 1883.

\bibitem{1985-JRM-TheTowersOfBrahmaAndHanoiRevisited-Wood}
D.~Wood.
\newblock The towers of {Brahma} and {Hano\"{\i}} revisited.
\newblock {\em Journal of Recreational Mathematics (JRM)}, 14(1):17--24, 1981.

\end{thebibliography}

\begin{LONG}
\medskip 
\textbf{Acknowledgements:}
 We would like to thank Claire Mathieu, Jean-Paul Allouche and Srinivasa Rao for corrections and encouragements, and  
%Pablo Pérez-Lantero, 
Javiel Rojas-Ledesma and 
Carlos Ochoa-Méndez 
for their comments on preliminary drafts.
\textbf{Funding:} J\'er\'emy Barbay is partially funded by the Millennium Nucleus RC130003 ``Information and Coordination in Networks''.
%
% \textbf{Author Contributions:} J\'er\'emy Barbay is the sole author of the presented work.
%
% \textbf{Competing Interests:} The author declares that he has no competing financial interests relevant to the material exposed in this article.
%
%\textbf{Data and Material Availability:}
\end{LONG}

\begin{SHORT}
\newpage
\end{SHORT}
\appendix
\begin{DISKPILEPROBLEM}
\section*{Appendix}
\section{Disk Pile Problem} \label{sec:diskPileProblem}

\newenvironment{solution}{\emph{Solution:}}{\qed}

\providecommand{\diskPile}{\textsc{Disk Pile}}
\providecommand{\diskPilePb}{\textsc{Disk Pile} problem}

The {\hanoitpb} is a classic example on recursivity, originally proposed by {\'E}douard Lucas~\cite{1883-MISC-LaTourDHanoi-Lucas} in 1883.  A recursive algorithm is known since 1892, moving the $n$ disks of a \hanoit\ in $2^n-1$ unit moves, this value being proven optimal by a simple lower bound~\cite{1892-BOOK-MathematicalRecreationsAndEssays-Ball}.  

Consider the {\diskPilePb}, a very simple variant where we allow some disks to be of the same size.  This obviously introduces some much easier instances, including an extreme one where the disks are all the same size and the resulting tower can be moved in linear time (see Figure~\ref{3diskPileInMove} for the sequence of steps moving such a tower of size $3$ with a single size of disks).
\renewcommand{\hanoiState}[3]{%
  \begin{tabular}{c}
%   \hspace{-3cm} #4    \\
  \framebox[66pt]{\vbox{%
      \hbox{%
        \makebox[22pt]{#1}%
        \makebox[22pt]{#2}%
        \makebox[22pt]{#3}%
        }%
      \hbox{%
        \makebox[22pt]{$\pegA$}%
        \makebox[22pt]{$\pegB$}%
        \makebox[22pt]{$\pegC$}%
        }%
      }  
    } 
  \end{tabular}
  }%

\begin{figure}[h] 
\begin{minipage}{\linewidth}\centering
  \hanoiState{\peg{\anonymea}{\anonymeb}{\anonymec}{}} {} {} %
  \hanoiState{\peg{\anonymea}{\anonymeb}{}{}} {} {\peg{\anonymec}{}{}{}} %
  \hanoiState{\peg{\anonymea}{}{}{}} {} {\peg{\anonymec}{\anonymeb}{}{}} %
  \hanoiState{\peg{}{}{}{}} {} {\peg{\anonymec}{\anonymeb}{\anonymea}{}} %
  \caption{Moving a \diskPile\ of size $3$.\label{3diskPileInMove}}
\end{minipage}
\end{figure}

\begin{enumerate}
\item Give a recursive algorithm to move a {\diskPile} from one peg to the other, using only one extra peg, knowing that $\forall i\in\{1,\ldots,s\}$, $n_i$ is the number of disks of size $i$.  Your algorithm must be efficient for the cases where all the disks are the same size, and where all the disks are of distinct sizes.

  \begin{solution}
  We present an algorithm in Figure~\ref{fig:DiskPileMovePython}. It is very similar to the algorithm moving a {\hanoit}, the only difference being that it moves the $n_i$ disks of size $i$ at the same time, in $n_i$ consecutive moves.
\begin{figure}[h]
\centering
\begin{INUTILE}
\begin{minipage}[t]{.45\linewidth}
\caption{$DiskPile.move(s,sizes[1..s],A,B,C)$\label{fig:move}}
\begin{lstlisting}
IF s>0 
 move(s-1,sizes[1..s-1],A,C,B)
 for i in [1..sizes[s]) 
  move(A,C)
 move(s-1,sizes[1..s-1],B,A,C)
ENDIF
\end{lstlisting}
\end{minipage}
\end{INUTILE}
\begin{minipage}[t]{.7\linewidth}
\caption{Python code to move a Disk Pile\label{fig:DiskPileMovePython}}
\begin{lstlisting}
def diskPileMove(n,sizes,a,b,c):
  if n>0 :
    move(n-sizes[-1],sizes[0:-1],a,c,b)
    for i in range(0,sizes[-1]):
      move(a,c)
    move(n-sizes[-1],sizes[0:-1],b,a,c)
\end{lstlisting}
\end{minipage}
\end{figure}
\end{solution}

\item Give and prove the worst case performance of your algorithm over all instances of fixed $s$ and vector $(n_1,\ldots,n_s)$.

  \begin{solution}
  By solving the recursive formula directly given by the recursion of the algorithm, one gets that the $n_s$ largest disks are moved once, the $n_{s-1}$ second largest disks are moved twice, the $n_{s-2}$ third largest disks are moved four times, and so on to the $n_1$ smallest disks, which are each moved $2^{s-1}$ times.  Summing all those moves give the number of moves performed by the algorithm:
    $$\sum_{i\in\{1,\ldots,s\}} n_i 2^{s-i}$$ 
Note that for $s=n$ and $n_1=\cdots=n_s=1$, this yields $\sum_{i=1}^{s-1} 2^i = 2^n -1 $, the solution to the traditional \hanoitpb.
  \end{solution}

\item Prove that a performance of $\sum_{i\in\{1,\ldots,s\}} n_i 2^{s-i}$ is optimal.

  \begin{solution}
  We prove a lower bound of $\sum_{i\in\{1,\ldots,s\}} n_i 2^{s-i}$, for $n$ disks of $s$ distinct sizes, with $n_i$ disks of size $i$ by induction on the number of types of disks.  We prove by induction on the number of types of disk $s$ that any pile of disks of sizes $(n_1,\ldots,n_s)$ requires $\sum_{i\in\{1,\ldots,s\}} n_i 2^{s-i}$ disk moves to be moved to another peg.
  \begin{itemize}
\item \emph{Initial Case}: for $s=1$ the bound is $n_1$ and is obviously true, since each disk must be individually moved from one peg to the other.
\item \emph{Inductive Hypothesis}: suppose there is some $\sigma\geq1$ so that any pile of disks sizes $(n_1,\ldots,n_\sigma)$ requires $\sum_{i\in\{1,\ldots,\sigma\}} n_i 2^{\sigma-i}$ disk moves to be moved to another peg.
\item \emph{Inductive Step}: consider a pile of disks of sizes $(n_1,\ldots,n_{\sigma+1})$: clearly all the disks of sizes smaller than ${\sigma+1}$ need to be gathered on a unique peg before the largest disks can be moved, to allow those last ones to be moved in $n_{\sigma+1}$ disk moves, after which all the disks of sizes smaller than ${\sigma+1}$ need to be stacked above the largest ones.  By the inductive hypothesis, moving the smaller disks will require $\mathbf{2}\sum_{i\in\{1,\ldots,\sigma\}} n_i 2^{\sigma-i}$ disk moves, to be added to the $n_{\sigma+1}$ disk moves.  Hence, any pile of disk of sizes $(n_1,\ldots,n_{\sigma+1})$ requires $\sum_{i\in\{1,\ldots,{\sigma+1}\}} n_i 2^{\sigma+1-i}$ disk moves to be moved to another peg.
\item \emph{Conclusion}: The inductive hypothesis is verified for the initial case where $s=1$, and propagates to any value of $s\geq1$ through the inductive step.  We conclude that any pile of disks of sizes $(n_1,\ldots,n_s)$ for $s\geq1$ requires $\sum_{i\in\{1,\ldots,s\}} n_i 2^{s-i}$ disk moves to be moved to another peg.
  \end{itemize}
  \end{solution}

\item What is the worst case complexity of the \diskPilePb\ over all instances of fixed value $s$ and fixed total number of disks $n$?

  \begin{solution}
  The worst case (of both the algorithms and the most precise lower bound with the number of disks of each size fixed) occurs when $n_1 = n-s+1$ and $n_2=\ldots=n_s=1$.  Using the previous results it yields a complexity of $2^{s-1}(n-s+1)+\sum_{i=1}^{s-1}2^i = 2^{s-1}(n-s+2)-1$ steps in the worst case over all instances of fixed value $s$ and fixed total number of disks $n$. This correctly yields $2^n-1$ when $s=n$, in the worst case over all instances of fixed total number of disks $n$.
  \end{solution}

\end{enumerate}

%%% Local Variables:
%%% mode: latex
%%% TeX-master: "2016-FUN-EducationalVariantsOfTheHanoiTowerProblem-Barbay"
%%% End:

\end{DISKPILEPROBLEM}

\begin{SHORT}
\section{Proofs omitted for lack of space}
\label{sec:omitted-proofs}

\subsection{Proof of Lemma~\ref{reciprocofconnexitylemma} (Connexity)}
\begin{proof}
By induction: 
for one disk it is true; 
for $k$ disks, if the \texttt{insertion} point after the \texttt{insertion} of disc $d$ is
above $d$ then add the larger and then the $k-1$ disks left,
else add the smaller and then the $k-1$ disks left.
\end{proof}

\subsection{Proof of Theorem~\ref{res:correctness} (Correctness)}

\begin{figure}\centering
\begin{minipage}[t]{.3\linewidth}
\caption{\\$move001(n,A,B,C)$\label{fig:move001}}
\begin{lstlisting}
def move001(n,a,b,c):
 if n == 1 :
  move(a,c)
 elif n>1 :
  move010(n-2,a,c,b)
  move(a,c)
  move(a,c)
  move001(n-2,b,a,c)
\end{lstlisting}
\end{minipage}
\end{figure}

\begin{proof}
Consider the property $\IH(n)=$ ``$\forall(x,y,z)\in\{0,1\}^3,$ $\forall i \leq n,$ $\move{xyz}{i}{A}{B}{C}$ is correct''.  $\IH(0)$ is trivially true, and $\IH(1)$ can be checked for all functions at once. For all values $x,y,z$, the function $\move{xyz}{1}{A}{B}{C}$ is merely performing the step $\unitarymove{A}{C}$. The hypothesis $\IH(1)$ follows.  Now, for a fixed $n> 1$, assume that $\IH(n-1)$ holds: we prove the hypothesis $\IH(n)$ separately for each function.

\begin{itemize}
% Correctness of MOVE 000
\item {Analysis of $\move{000}{n}{A}{B}{C}$:}
\begin{enumerate}
\item According to $\IH(n-1)$ the call to $\move{100}{n-1}{A}{B}{C}$ is correct if $(i)$ and $(p)_{100}$ are respected.  $(i)$ is implied by $(i)$ on $\move{000}{n-1}{A}{B}{C}$; $(p)_{100}$ is implied by $(p)_{000}$ and the remaining disk on $A$ ($a-n \mod 2 \equiv 0 \Rightarrow a-(n-1)  \mod 2 \equiv 1 \mod 2$).
  % According to $\IH(n-1)$ the call to $\move{100}$ is correct, so it moves effectively $n$ disks from $A$ to $B$.  Note that precondition $(i)$ is satisfied by the requirement for $\move{000}$, and that the requirement $(p)_{100}$ on the parity of the number of disks on $A$ is satisfied because of the remaining disk on $A$, while $\IH(n-1)$ gives $(p)_{000}$.
\item The step $\unitarymove{A}{C}$ is possible and legal because of the precondition $(i)$ for $\move{000}{n}{A}{B}{C}$: the disk moved was in the $n$ first removed from $A$, and so can be introduced on $C$.
\item The call to $\move{001}{n}{A}{B}{C}$ is symmetrical to $1$, and so correct.
\item We can check the final state by verifying that the number of disks removed from $A$ and added to $C$ is $(n-1) + 1 = n$.
\end{enumerate}
So $\move{000}{n}{A}{B}{C}$ is correct.

% Correctness of MOVE 100
\item {Analysis of $\move{100}{n}{A}{B}{C}$:}
\begin{enumerate}
  \item $\move{100}{n-2}{A}{B}{C}$ is correct according to $\IH(n-1)$, as the
  requirements are also:
  The requirement $(i)$ is given by $(i)$ for the initial call, 
  and the parity $(p)_{100}$ is respected because we move two disks 
  less than in the current call to $\move{100}{n}{A}{B}{C}$.  
%  So this call $\move{100}(n-2)$ is correct and moves $n-2$ disks to $B$
\item The two disks left (let us call them $\alpha$ and $\beta$) are in position (given fig. \ref{TwoLastDisksRemoved}, $(i)$) such that the removal order on $A$ is $(\alpha,\beta)$ and the \texttt{insertion} order on $C$ is $(\beta,\alpha)$ (see fig.\ref{TwoLastDisksRemoved}, $(ii)$).  They can be inserted on $C$ because of requirement $(i)$.  So the two disks are correctly moved in two steps.

%  the first one extracted is the smallest (because the number of disks
%  on the peg is odd) and the second disk follows.
  
%   On peg $C$ after the insertion of the disk $\alpha$, the insertion
%   point is below as the number of disk is then odd: so we can insert
%   the disk $\beta$.  The two disks are well placed one compared to the
%   other, and also compared to other disks because of precondition
%   $(i)$.
  
\item The requirements for $\move{010}{n-2}{A}{B}{C}$ are satisfied:
  \begin{itemize}
\item{$(i)$} stand as a consequence of the precondition $(i)$ for the current call, as the $n-2$ disks to be moved on $C$ were on $A$ before the original call, in the middle of $\alpha$ and $\beta$.
\item{$(p)_{010}$}: The number of disks on $C$ is still even as we added two disks. The number of disks on $A$ is still odd as we removed two disks.
  \end{itemize}
  So, because of $\IH(n-2)$, $\move{010}{n-2}{A}{B}{C}$ is correct.
\end{enumerate}
So $\move{100}{n}{A}{B}{C}$ is correct.

\begin {figure}[h]
\begin{center}
\parbox{5cm}{
% \begin{center}  \input{BouncingTowers/twoDisksOnOddPeg.pstex_t}\end{center}
 \begin{center}  \includegraphics{twoDisksOnOddPeg}\end{center}
\begin{center}$(i)$\end{center}
  $n$ odd: $a$ is removed first,\\
  $y$ is removed second.
}
\parbox{5cm}{
% \begin{center}  \input{BouncingTowers/twoDisksOnEvenPeg.pstex_t}\end{center} 
  \begin{center}
  \includegraphics{twoDisksOnEvenPeg}
  \end{center}
\begin{center}$(ii)$\end{center}
  $n$ even: $y$ is removed first,\\
  $x$ is removed second.
}
\end{center}
\caption{Removal order of the last two disks.\label{TwoLastDisksRemoved}}
\end {figure}

\item {Analysis of $\move{001}{n}{A}{B}{C}$:} This function is the exact symmetric of $\move{100}{n}{A}{B}{C}$, for a task exactly symmetric, so has a symmetric proof of its correctness.

\item {Analysis of $\move{010}{2}{A}{B}{C}$:} The two disks (let us call them $\alpha$ and $\beta$) are in position (given fig.  \ref{TwoLastDisksRemoved}, $(ii)$) such that the \texttt{removal} order on $A$ is $(\beta,\alpha)$ and the \texttt{insertion} order on $C$ is $(\alpha,\beta,)$, as $A$ and $C$ have the same parity.  $\beta$ can be inserted on $B$ and they can both be inserted on $C$ because of requirement $(i)$.  So the two disks are correctly moved in three steps, using peg $B$ to dispose temporally disk $\beta$.  So $\move{010}{2}{A}{B}{C}$ is correct.

\item {Analysis of $\move{010}{n}{A}{B}{C}$ if $n>2$:} All along of this proof of correctness we shall use the fact that fixing $2$ disks on the same peg doesn't change the parity of this peg.
\begin{enumerate}
\item $\move{010}{n-2}{A}{B}{C}$ is correct as: from $(i)$ for the initial call results $(i)$ for the first recursive call; $(p)_{010}$ is a natural consequence of $(p)_{010}$ for the initial call (because parity conserved when icing two disks).  So $\IH(n-1)$ implies that $\move{010}{n-2}{A}{B}{C}$ is correct.
\item $A$ and $B$ having different parities, we can move two consecutive disks in two consecutive calls as for $\move{100}{n}{A}{B}{C}$.
\item The second recursive call to $\move{010}{n-2}{A}{B}{C}$ verifies conditions $(i)$ and $(p)_{010}$ as only two extremes disk have been removed from $A$.
\item The two next steps are feasible because of the difference of parity between $B$ and $C$ (same argument as point $2$).
\item The last recursive call is symmetric to the first call, as we move back the $n-2$ disks between the two extreme disk, but this time on $C$.
\end{enumerate}
So $\move{010}{n}{A}{B}{C$} is correct. \qedhere
\end{itemize}
\end{proof}

\subsection{Proof of Theorem~\ref{res:proof-optimality} (Optimality)}

\begin{proof}

\begin{TODO}
CHANGE definition of Induction hypothesis to have the forall inside.
\end{TODO}
Define the induction hypothesis $\IH(n)$ as ``$\forall(x,y,z)\in\{0,1\}^3\, \move{xyz}{n}{A}{B}{C}$ moves optimally $n$ disks from $A$ to $C$''. Trivially $\IH(0)$ and $\IH(1)$ are true.  Suppose that there exists an integer $N>1$ such that $\forall n<N$, the induction hypothesis $\IH(n)$ is true.  We prove that $\IH(N)$ is then also true.

\begin{itemize}

\item $\move{000}{N}{A}{B}{C}$ is optimal:

%(please see Figure~\ref{subdivisionOfTotalGraphForHanoi})
$\move{000}{N}{A}{B}{C}$ for $N>0$ consists of one
call to $\move{100}{N}{A}{C}{B}$, one unitary step,
and one call to $\move{001}{N}{B}{A}{C}$.

So it moves optimally (by $\IH(N-1)$) from $\A A\ldots A$ to $\A
B\ldots B$, and then to $\C B\ldots B$, and after that to $\C C\ldots C$.
(In Figure~\ref{TotalGraphForHanoi} the right edge of
the triangle.)

A path not going through states $\A B\ldots B$ or $\C B\ldots B$ would take more steps:
\begin{itemize}
\item if we don't go through the state $\A B\ldots B$, then the state $\A C\ldots C$ is necessary, with a cost of $f_{100}(N-1)$, and also the state $\B C\ldots C$ (with a cost of $1$), and at the end of the path we have to go through the state $\C A\ldots A$, which optimal path to go to the final $\C C\ldots C$ state is of length $f_{100}(N-1)$: this path is of length $f_{100}(N-1)+1+f_{100}(N-1)$ and is already as long as the one given by $\move{000}{N}{A}{B}{C}$.
\item if we go through $\A B\ldots B$, but not through $\C B\ldots B$, then the path is not optimal as it must go through $\A C\ldots C$ and the optimal path from $\A A\ldots A$ to $\A C\ldots C$ doesn't go through $\A B\ldots B$.
\end{itemize}
So $\move{000}{N}{A}{B}{C}$ is optimal.

\item $\move{100}{N}{A}{B}{C}$ is optimal:

  $\move{100}{N}{A}{B}{C}$ for $N>1$ consists of one call to $\move{100}{N-2}{A}{C}{B})$, two steps, and one call to $\move{010}{N-2}{B}{A}{C})$.

As before, we shall consider these recursive calls of order smaller
than $N$ as optimal because of $\IH(N-2)$. So we know how to move
optimally 
from    $\A A A\dots A A$ 
to      $\A A B\dots B A$, 
to      $\A C B\dots B A$, 
then to $\A C B\dots B C$ 
and to  $\A C C\dots C C$
(in figure \ref{totalGraphDecompositionWithOneDiskIcedA}), this
corresponds to the left edge of the triangle).

We must now prove that other paths take more steps:
\begin{itemize}
\item We cannot avoid the state $\A C B\dots B C$, neither $\A C
  B\dots B A$, as there is no other way out of $\A C C\dots C C$.
\item if we avoid the state $\A A B\dots B A$ then the optimal path to
  $\A C B\dots B A$ necessarily passes by $\A BA\ldots AA$ and $\A
  CA\ldots AA$, and is of length
  $f_{100}(N-2)+1+f_{010}(N-2)+1+f_{010}(N-2)$, which is longer than
  the whole solution given by the algorithm, of length $f_{100}(N)=
  f_{100}(N-2) + 2 + f_{010}(N-2)$.
\end{itemize}
So $\move{100}{N}{A}{B}{C}$ is optimal.

\item $\move{010}{N}{A}{B}{C}$ is optimal:

$\move{010}{1}{A}{B}{C}$ and $\move{010}{2}{A}{B}{C}$ are special cases, 
we can see in graphs $G'(1)$ and $G'(2)$ on figure 
\ref{smallGraphs} page \pageref{smallGraphs} 
that the optimal paths between $\A B\ldots B$ and $\A C\ldots C$
are of length $1$ and $3$, as the solutions produced by the algorithm.
So $\move{010}{1}{A}{B}{C}$ and $\move{010}{2}{A}{B}{C}$ are proven optimal.

$\move{010}{N}{A}{B}{C}$ for $N>2$ corresponds to a
path going through the states (the first disk being fixed on $\B$):
(please report to fig. \ref{totalGraphDecompositionWithOneDiskIcedA} 
from $\A C C\ldots CC$ to $\A B B\ldots BB$ down left to down right. )
 \[
 \A C C\dots C C
\stackrel{f_{010}(N-2)}{\longrightarrow}
 \A C B\dots B C
\stackrel{2}{\longrightarrow}
 \A A B\dots B C
\]\[
\stackrel{f_{010}(N-2)}{\longrightarrow}
 \A A C\dots C C 
\stackrel{2}{\longrightarrow}
 \A B C\dots C B
\stackrel{f_{010}(N-2)}{\longrightarrow}
 \A B B\dots B B 
\]

% \begin{figure}[h]
% \parbox{5cm}{
% \[\begin{array}{ll}
% \B A A\dots A A\\
% &\downarrow f_{010}(N-2)\\
% \B A C\dots C A\\
% &\downarrow 2\\
% \B B C\dots C A\\
% &\downarrow f_{010}(N-2)\\
% \B B A\dots A A\\
% &\downarrow 2\\
% \B C A\dots A C\\
% &\downarrow f_{010}(N-2)\\
% \B C C\dots C C
% \end{array}\]
% }
% \parbox{8cm}{
% \input{BouncingTowers/totalGraphEtiquettedWithOneDiskFixedB.pstex_t}
% }
% \caption{Total Graph Decomposition with one fixed disk.
% \label{totalGraphDecompositionWithOneDiskIcedB}}
% \end{figure}

We shall demonstrate that all other paths take more steps:
\begin{itemize}
\item The states $\B A C\dots C A$ and  $\B C A\dots A C$ are mandatory,
  for connexity, and so are  $\B A C\dots C B$ and  $\B C A\dots A B$.
\item if we go through $\B B C\dots C B$, 
then it's $\B B A\dots A B$ which is mandatory.
\item if we contourn $\B B C\dots C B$, then we shall go through 
$\B A B\dots B B$, $\B C B\dots B B$ and $\B C A\dots A B$:
the total path would be of length $3+4 f_{010}(N-2)$,
to be compared with $4 +3 f_{010}(N-2)$ (We trade one step
with one recursive call).
As $f_{010}(N-2) \geq 1$ for $N-2\geq q$ (i.e. $N\geq 3>2$),
$\move{010}{N}{A}{B}{C}$ is optimal for $N>2$.
\end{itemize}
So $\move{010}{N}{A}{B}{C}$ is optimal. \qedhere
\end{itemize}
\end{proof}

\end{SHORT}

%%% Local Variables:
%%% mode: latex
%%% TeX-master: "2016-FUN-EducationalVariantsOfTheHanoiTowerProblem-Barbay"
%%% End:

\end{document}